\documentclass[noshowpacs,preprintnumbers,amsmath,amssymb,superscriptaddress,onecolumn]{revtex4}

\usepackage{amsmath}
\usepackage{eqnarray}
\usepackage{amssymb}
\usepackage{mathptmx}
\usepackage{color}
\usepackage{caption}
\usepackage{graphicx}
\usepackage{esvect}

\newcommand{\be}{\begin{equation}}
\newcommand{\ee}{\end{equation}}

\newtheorem{Theorem}{Theorem}
\newtheorem{proposition}{Proposition}
\newtheorem{corollario}{Corollary}
\newtheorem{Definition}{Definition}
\newtheorem{definition}{Definition}
\newtheorem{remark}{Remark}
\newenvironment{proof}{\textbf{Proof~}}{\hfill$\square$}

\begin{document}



\title{Free energies of Boltzmann Machines: self-averaging, \\ annealed and replica symmetric approximations in the thermodynamic limit}

\author{Elena Agliari}
\affiliation{Dipartimento di Matematica, Sapienza Universit\`a di Roma, Italy\\
GNFM-INdAM Sezione di Roma, Italy}
\author{Adriano Barra}
\affiliation{Dipartimento di Matematica e Fisica, Universit\`a del Salento, Lecce, Italy\\
GNFM-INdAM Sezione di Roma, Italy\\
INFN Sezione di Lecce, Italy}
\author{Brunello Tirozzi}
\affiliation{Dipartimento di Fisica, Sapienza Universit\`a di Roma, Italy\\
Enea Research Center, Frascati, Italy}


\begin{abstract}
Restricted Boltzmann machines (RBMs) constitute one of the main models for machine statistical inference and they are widely employed in Artificial Intelligence as powerful tools for (deep) learning. However, in contrast with countless remarkable practical successes, their mathematical formalization has been largely elusive: from a statistical-mechanics perspective these systems display the same (random) Gibbs measure of bi-partite spin-glasses, whose rigorous treatment is notoriously difficult.
\newline
In this work, beyond providing a brief review on RBMs from both the learning and the retrieval perspectives, we aim to contribute to their analytical investigation, by considering  two distinct realizations of  their weights (i.e., Boolean and Gaussian) and studying the properties of their related free energies. More precisely, focusing on a RBM characterized by digital couplings, we first extend the Pastur-Shcherbina-Tirozzi method (originally developed for the Hopfield model) to prove the self-averaging property for the free energy, over its quenched expectation, in the infinite volume limit, then we explicitly calculate its simplest approximation, namely its annealed bound.
\newline
Next, focusing on a RBM characterized by analogical weights, we extend Guerra's interpolating scheme to obtain a control of the quenched free-energy under the assumption of replica symmetry (i.e., we require that the order parameters do not fluctuate in the thermodynamic limit): we get self-consistencies for the order parameters (in full agreement with the existing Literature) as well as the critical line for ergodicity breaking that turns out to be the same obtained in AGS theory.
As we discuss, this analogy stems from the slow-noise universality.
\newline
Finally, glancing beyond replica symmetry, we analyze the fluctuations of the overlaps for a correct estimation of the (slow) noise affecting the retrieval of the signal, and by a stability analysis we recover the Aizenman-Contucci identities typical of glassy systems.
\end{abstract}
	
\maketitle
	
	
\section{Introduction: Boltzmann machines in a nutshell.}
Boltzmann machines (BMs) play a key role in Artificial Intelligence: being able to learn internal representations (when fed by external data) and to solve
difficult combinatoric problems (when suitably trained), they can be efficiently employed in machine learning and machine statistical-inference. Also, extensions of the BMs, i.e., the so-called Deep Boltzmann machines \cite{Hinton1,cinese}, allow for Deep Learning, that is the novel generation of Artificial Intelligence. Their applications are broadly ranged in Science (from Particle Physics \cite{PP} to Computational Biology \cite{CB}), not to mention
the applied world of technology, where their usage has become pervasive \cite{DL1,Martens}.
Further, several models of memory formation and pattern recognition in Theoretical Immunology (regarding the adaptive branch of the immune system of
mammals) \cite{Agliari4,Agliari5,Annibale1,Annibale2,Mozeika1,Mozeika2} are naturally framed in the mathematical scaffold of BMs.

In a nutshell, a BM is a network of symmetrically connected units (also called neurons or spins) divided into two or more layers. Here we shall focus
on \emph{restricted} Boltzmann machines (RBMs), made of two layers called ``visible'' and ``hidden'', respectively; the units are connected to each
other across layers, but there is no intra-layer communication (this is the {\em restriction} in a RBM, see Fig.~\ref{fig:uno}, left panel).
\newline
As anticipated, RBMs can be used to solve quite different computational problems \cite{BM-ML2,Huang1,Huang2,BM-ML1,Decelle1,Florent3}.
In order to achieve this, the machine has first to undergo a training process where its parameters, namely the thresholds $\theta$ for neuron firing (i.e., for spin flip in the binary case \cite{Amit}) and the weights $\xi$ associated to links (i.e., the synaptic values), are stochastically tuned according to proper algorithms (e.g., contrastive divergence \cite{Hinton1} and simulated annealing \cite{kirkpatrick}). {After this process, one can initialise the visible layer in a given state (input) and make the system evolve towards its ground state}\footnote{The symmetry of the interactions between neurons ensures the detailed balance which, in turn, ensures the relaxation to a Gibbs measure.}; the latter (output) provides the solution of the problem \cite{tank}.
\newline
{Let us explain this concept more extensively. During the training stage, a set of noisy data (assumed to be {independently and identically generated by a given probability distribution}) is shown to the visible layer of the RBM. This is obtained by encoding the data in terms of state configurations for the visible layer (e.g., data can be constituted by pictures, where each pixel, either black or white, is associated to the state, either $+1$ or $-1$, of the related binary visible unit). For each item presented to the visible layer, the system  relaxes to the pertinent equilibrium distribution and, once a stable thermalization is reached, an update in the machine parameters $(\xi, \theta)$ plays as a {\em learning rule} in such a way that the Kullback-Leibler divergence between the internal {probabilistic} representation of the data (given by the equilibrium distribution of the system) and the external distribution generating these data is minimized.
In this process the hidden units, whose states are not specified by the data provided, act as latent variables and allow the RBM to store information about the learnt data. Each hidden unit turns out to be associated to a ``feature'' (e.g., the extent of correlations among the entries of input data) and the set of all the features is supposed to statistically characterise the data provided \cite{Bengio1}. Actually, the nature of such features emerges during the training without control from the user.
After training, if a new item (from the same distribution and possibly noisy) is shown, the machine will be able to ``recognize'' its statistical structure.
For instance, if trained on pictures, the machine will theoretically model the distribution of pictures and can then be used for reconstruction tasks: given a partial picture as input, the machine will use the internal model to provide, as output, the complete picture.
We refer to Sec.~\ref{sec:overview} for a more technical explanation.}

Despite the success in practical applications, a rigorous control of RBMs remains a hard problem from the mathematical perspective. {It is worth recalling that, in the mathematical investigations we are interested in, weights on links are assumed as quenched and extracted according to a given distribution which is meant as the result of some training process. Thus, the RBMs we are dealing with can} be
classified as two-party spin-glasses \cite{Barra1,Barra4,Guerra0,MPV,Panchenko}. As a consequence, an extensive use of techniques and tools from the Statistical Mechanics of
Disordered Systems is in order and, quoting Talagrand, this is still a {\em challenge for mathematicians} \cite{Tala0}.
Remarkably, reaching a full, rigorous description for these models is further hindered by the fact that, due to the wide range of their applications,
there exists a number of variations on theme \cite{DL1} and their differences are subtle but important. These differences mainly affect the nature of links
(e.g., Gaussian or Boolean) and the nature of the units (e.g., Gaussian or Boolean) of the hidden layer \cite{Barra0}.

As a matter of fact, only a few rigorous results on this subject are available (see e.g., \cite{Benarous,Martens}) and there is a urgent and significant gap to fill, as broadly recognized (see e.g., \cite{DL1,Benarous}).
A possible way to pave towards rigorous advances has recently been investigated \cite{Agliari1,Barra3,Florent1,Florent2,Mezard,Monasson,Dani} and it is based on the thermodynamic equivalence of RBMs and Hopfield neural networks \cite{Barra2}. The latter constitute the standard model for associative memory \cite{Hopfield},
for which several mathematical tools have been made available along the past two decades (see e.g., \cite{Barra5,Bovier2,Bovier4,Tala1,Tala2,Tirozzi1,Tirozzi2}).
In the large volume limit, this bridge allows extending several rigorous mathematical approaches, originally developed within the framework of
spin-glasses and associative neural networks, to machine learning as well. In particular, in this paper, we will focus on two variations on this theme and address their
properties with two different techniques as summarized in the following:
\begin{itemize}
\item {\em RBMs with Boolean links}.
\newline
This is a bi-partite spin-glass where the two parties (i.e., the two layers) contain variables of different nature: the visible layer is made of binary Ising spins, while the hidden layer is made of real Gaussian spins. The weights on links connecting the units belonging to different layers are binary {($\pm 1$)} and extracted independently and identically with equal probability. For this model we adapt the Pastur-Shcherbina-Tirozzi scheme, originally developed for the standard Hopfield model \cite{Tirozzi2,Tirozzi3}. This procedure allows us to prove the self-averaging of the main observable, namely the free-energy of the model, over its quenched expectation.  Next, we provide its first approximation, namely its annealed bound, which will be achieved via Jensen inequality.
\newline
The underlying idea in this route is to construct a {\em matryoshka} of $\sigma$-algebras  generated by the random weights associated to links and to define conditional expectations over them. Next, we show that the sequence of conditional expectations of the extensive free-energy
constitutes a martingale and, exploiting its properties (mainly Doob Theorem \cite{martingale}), we can bound fluctuations in the intensive free-energy and prove that it converges to its quenched expectation.

\item {\em RBMs with Gaussian links}.
\newline
This is a bi-partite spin-glass where the two parties (i.e., the two layers) contain variables of different nature: the visible layer is made of binary Ising spins, while the hidden layer is made of real Gaussian spins. The weights on links connecting the units belonging to different layers are real valued and sampled identically and independently from a Gaussian distribution $\mathcal{N}[0,1]$.
\newline
In order to get an explicit expression of the infinite volume limit of the quenched free-energy for these machines we adapt an interpolation scheme, originally developed by Guerra for the Sherrington-Kirkpatrick spin-glass \cite{Mingione}.
Here the underlying idea is to compare the original network, where the two parties interact extensively, with an effective model where the two parties are no longer interacting, rather, they experience random external fields, whose statistical distributions mirror the real ones, namely each neuron feels an external field which mimics the internal field generated by neurons it is connected with. This makes the calculations feasible because the effective model is one-body thus, in principle, integrable. As standard practice in neural networks \cite{Amit,Coolen}, we assume that, in the infinite volume limit, the order parameters do not fluctuate (their distributions get $\delta$-peaked over their thermodynamic values), namely we will assume replica symmetry.
\newline
Since the slow noise generated by the not-retrieved patterns is \emph{universal} (namely, in the thermodynamic limit, the replica-symmetric expression for the noises stemming from pattern entries sampled from a Bernoullian distribution and from a Gaussian one coincide), the results of this section naturally extend to the machine with digital links investigated in the previous one. This remark prompts to a stability analysis, where we investigate overlap fluctuations (hence glancing beyond replica symmetry) and we find, as standard in glassy systems, that these fluctuations obey the (suitably adapted) Aizenman-Contucci constraints.
\end{itemize}

{\em Summary of the paper:} In the next section we present our perspective on the way RBMs actually work, while we leave the next Secs.~\ref{sec:BT}-\ref{sec:FG} to novel and rigorous results, Sec.~\ref{sec:signal} for deepening these results from an applied perspective and, finally, the last section contains our conclusions. Particularly long proofs (those of Theorem 1, Proposition 2, Theorem 3 and Corollary 1) are left in the Appendices to lighten the exposition, while all the other proofs are brief enough to be included in the main text.
\newline
Just for the next section we will relax the mathematical rigour in order to allow for a minimal but fluent presentation and provide a streamlined review.

\section{A quick overview on restricted Boltzmann machines} \label{sec:overview}

As anticipated in the previous section, a RBM is a two-party system, where one party (or {\em layer} to keep the original jargon \cite{Hinton1,DLbook}), referred to as visible, receives {input data} from the outside world, while the other party (or layer), referred to as hidden, is dedicated to figure out correlations in {these input data}.
Typically, a set of {$M>0$} data vectors {$\{ \sigma_1, \sigma_2, ..., \sigma_M \}$} (i.e., the so-called {\em training set} \cite{DLbook}) is presented to the machine and, under the assumption that these data have been generated by the
same probability distribution $Q(\sigma)$, the ultimate goal of the machine is to make an inner representation of $Q(\sigma)$, say $P(\sigma {|\xi, \theta})$, that is as close as possible to the
original one, i.e., it has to reconstruct the signal. Clearly, in order to get a good representation, the more complicated $Q(\sigma)$, the larger the training set\footnote{Actually, in order to optimize the training stage, one should also properly set the internal parameters of the machine such as the ratio between the sizes of the visible and hidden layer, the kind of the neurons, etc. \cite{Barra0,DLbook}.}.
\newline
Each layer is composed by spins (also called neurons in this context), $N$ for the visible layer and $P$ for the hidden layer, and these spins can be chosen with high generality, ranging from discrete-valued
(e.g., Ising spins), to real-valued (e.g., Gaussian spins).
The thermodynamic limit of the ratio between the layer sizes, denoted as $\alpha = \lim_{N\to \infty} P/N$, is a control parameter and usually one splits the case $\alpha = 0$ (possibly yielding to under-fitting) and the case $\alpha \in \mathbb{R}^+$ (possibly yielding to over-fitting) \cite{DLbook}, the latter being mathematically much more challenging.
\newline
Analogously, the entries of the weight matrix can be either real or discrete. Generally speaking, continuous weights allows {for
learning rules (e.g., the contrastive divergence involving weight derivatives) which are more powerful than} their discrete counterparts (the typical learning rule for binary weights is the Hebbian one \cite{Amit})
and are therefore more convenient during the learning stage; on the other hand,
binary weights are more performing in the so-called retrieval phase, that is, once the machine has learnt and is ready to {perform the task it has been trained for}.
This trade-off gave rise to a number of variations on theme within the world of RBMs, where the extremal cases are probably given by a machine with
binary (i.e., Boolean) versus real (i.e., Gaussian) weights, equipped with a binary visible layer and a real hidden layer: in the present work we will focus on both these cases and we will try to highlight equivalences (but also crucial differences) among these extrema.

Before presenting in details these cases, it is useful to summarize the mechanisms underlying the functioning of a standard RBM and, to this aim, we now introduce its effective Hamiltonian {(or ``cost function'' in a machine learning jargon)} as
\be
H_N(\sigma,z|\xi,\theta)= -\frac{1}{\sqrt{N}}\sum_{i=1}^{N} \sum_{\mu=1}^P \xi_i^{\mu}\sigma_i z_{\mu}-\sum_{i=1}^N \theta_i \sigma_i,
\ee
where $\sigma_i$ ($i \in [1, ..., N]$) denotes the state of the $i$-th visible unit, $z_{\mu}$ ($\mu \in [1, ..., P]$) denotes the state of the $\mu$-th hidden unit, $\xi_i^{\mu}$ denotes the weight associated to the link connecting the neurons labelled $i$ and $\mu$, and the factor $1/\sqrt{N}$ ensures the linear extensivity of the Hamiltonian with respect to the system volume \cite{Barra6}. The scalars $\theta_i$ ($i \in [1, ..., N]$) {can be interpreted as external fields acting on the visible units and provide thresholds for neuron firing: given a certain internal field $\sum_{\mu} \xi_i^{\mu} z_{\mu} / \sqrt{N}$ over $\sigma_i$, the larger $\theta_i$ and the more likely for the $i$-th neuron to fire, namely to be in an active state $\sigma_i=+1$.}
\newline
Now, this system is made to evolve by applying algorithms mimicking cognitive processes \cite{BM1,Hopfield}. More precisely, one splits {\em cognition} into two separate acts, namely distinguishing between {\em learning} (information) and {\em retrieval} (of the learnt information).
The former occurs on a slower time scale and implies a synaptic dynamics which is modeled by properly rearranging the set of weights and thresholds.
The latter occurs on a faster time scale and implies a neuronal dynamics which is modeled by properly rearranging the spin configuration, while keeping the weights quenched.
Given the gap between the time scales characterizing these dynamical processes\footnote{{The time scale for neuronal spikes is order of $50$ ms, while the time scale for synaptic rearrangement is order of hours and it takes order of weeks to consolidate}.}, one can treat them adiabatically, as done in the following subsections: the next one is dedicated
to synaptic dynamics (i.e., rearrangement of the weights), while the successive one to neural dynamics (i.e., rearrangement of the spins).

\subsection{A brief digression on slow variable's dynamics: learning}
In this subsection we focus on the algorithms underlying the learning stage and which imply the dynamic of weights (we refer to \cite{Coolen} for a more extensive treatment).
As mentioned in the beginning of Sec.~\ref{sec:overview}, the goal is to obtain an inner representation $P(\sigma|{\xi,\theta})$ which approximates $Q(\sigma)$; this is usually achieved by the minimization of the Kullback-Leibler cross entropy $D(Q,P)$, defined as
\be\label{KL}
D(Q,P) \dot{=} \sum_{\sigma}Q(\sigma)\ln \left[ \frac{Q(\sigma)}{P(\sigma)}\right],
\ee
where the sum runs over all the possible configurations of the visible layer and we have dropped the dependence on the parameters $(\xi,\theta)$ of $P(\sigma|\xi,\theta)$ to lighten the notation. To the same purpose we also introduce $\tilde{\xi}_i^{\mu} \dot{=} \xi_i^{\mu}/\sqrt{N}$, $\forall i, \mu$.
Notice that $D(Q,P)$ is minimal (and equal to zero) {if and only if} $P(\sigma)$ and $Q(\sigma)$ are identical. Now, by updating the weights {and the thresholds} by a gradient descent rule
\begin{eqnarray}
\Delta \tilde{\xi}_i^{\mu} &=& - \epsilon \frac{\partial D(Q,P)}{\partial \tilde{\xi}_i^{\mu}},\\
\Delta \theta_i &=& - \epsilon \frac{\partial D(Q,P)}{\partial \theta_i},
\end{eqnarray}
where $\epsilon$ is a small parameter {(also called learning rate)}, we get
\begin{eqnarray} \label{contrasto}
\Delta D(Q,P) &=& \sum_{i,\mu} \frac{\partial D(Q,P)}{\partial \tilde{\xi}_i^{\mu}}\Delta \tilde{\xi}_i^{\mu} +
\sum_{i} \frac{\partial D(Q,P)}{\partial \theta_i}\Delta \theta_i\\ \nonumber
&=& - \epsilon \left[ \sum_{i,\mu}\left( \frac{\partial D(Q,P)}{\partial \tilde{\xi}_i^{\mu}} \right)^2  {+}  \sum_{i}\left( \frac{\partial D(Q,P)}{\partial \theta_i} \right)^2 \right] \leq 0,
\end{eqnarray}
that is the cross-entropy $D(Q,P)$ decreases monotonically until a stationary state is reached (which, still, does not necessarily correspond to $D(Q,P)=0$). Now, in order to make this learning rule an explicit, operational, algorithm a bit of work is still necessary.
A key point is that weights in the BM are symmetric (i.e., they are undirected) and this, for (non-pathologic) stochastic dynamics, implies {\em detailed balance} which, in turn,
ensures that the invariant measure is the Gibbs one given by
\begin{equation}
P(\sigma,z) = \frac{e^{-\beta H_{{N}}(\sigma,z|{\xi,\theta})}}{Z_{P,N}(\beta| {\xi,\theta})},
\end{equation}
where
${Z_{P,N}(\beta|\xi,\theta)}$ is a normalization factor (or ``partition function'' in a Statistical Mechanics jargon \cite{Amit,Coolen}) and $\beta \in \mathbb{R}^+$ encodes for the noise intrinsically present in real data sets (in Physics $\beta$ plays as an inverse temperature, in proper units). {Now, marginalizing $P(\sigma, z)$ over the hidden layer $z$, we get $P(\sigma)$}.
Therefore, the {internal representation of the probability distribution is formally known} and this allows the
construction of explicit learning algorithms, among which the {\em contrastive divergence} that we are going to derive is probably the most applied \cite{DLbook}.
In order to proceed with the construction of a learning algorithm we explicitly define
\begin{eqnarray}
\label{eq:ZN1}
Z_{P,N}(\beta| \xi, \theta) &=& \int_{-\infty}^{+\infty}\prod_{\mu=1}^{P}d \mu (z_{\mu}) \sum_{\sigma}  e^{-\beta H_N(\sigma,z|\xi,\theta)},\\
\label{eq:ZN2}
Z_{P,N}(\beta | \sigma, \xi, \theta) &=& \int_{-\infty}^{+\infty}\prod_{\mu=1}^{P}d \mu (z_{\mu}) e^{-\beta H_N(\sigma,z|\xi ,\theta)},\\ \nonumber
\label{eq:ZN3}
P(\sigma) &=& \int_{-\infty}^{+\infty}\prod_{\mu=1}^{P}d \mu (z_{\mu}) P(\sigma, z) = \int_{-\infty}^{+\infty}\prod_{\mu=1}^{P}d \mu (z_{\mu}) \frac{e^{-\beta H_N(\sigma, z| \xi, \theta)}}{Z_{P,N}(\beta| \xi , \theta)} \\
\label{eq:ZN5}
&=& \frac{Z_{P,N}(\beta|\sigma, \xi, \theta)}{Z_{P,N}(\beta | \xi, \theta)}, \\ \nonumber
\label{eq:ZN4}
P(z| \sigma) &=& \frac{P(\sigma, z)}{P(\sigma)} = \frac{e^{- \beta H_N(\sigma,z|\xi, \theta)}}{Z_{P,N}(\beta|\sigma, \xi, \theta)},
\end{eqnarray}
where, summations are meant over all possible spin configurations and $d \mu (z_{\mu})$ is the Gaussian measure ($d\mu (z_{\mu}) = \exp(-z_{\mu}^2 \beta/2) \sqrt{\beta / (2 \pi)}$, for $\mu=1, ..., P$).
Thus, $Z_{P,N}(\beta | \xi, \theta)$ is the partition function of a system where both variable sets are free to evolve, while $Z_{P,N}(\beta | \sigma, \xi, \theta)$ is the partition function of a system where the visible layer is ``clamped'', namely forced to be in the configuration $\{\sigma\}$ encoded by one of the input data. Also, $P(\sigma)$ is the marginalized distribution and $P(z|\sigma)$ is the distribution for the configuration of the hidden layer being the visible layer clamped.
At this point we have to evaluate each single term inside (\ref{contrasto}):
{
\begin{eqnarray}
\nonumber
\frac{\partial D(Q,P)}{\partial \tilde{\xi}_i^{\mu}} &=& - \sum_{\sigma} Q(\sigma) \frac{\partial \ln P(\sigma)}{\partial \tilde{\xi}_i^{\mu}}\\
\nonumber
&=& - \sum_{\sigma} Q(\sigma) \frac{\partial }{\partial \tilde{\xi}_i^{\mu}} \left(\ln Z_{P,N}(\beta| \sigma, \xi, \theta) - \ln Z_{P,N}(\beta|\xi, \theta)  \right)\\
\nonumber
&=& \beta \sum_{\sigma} Q(\sigma)\Bigg(\int_{-\infty}^{+\infty}\prod_{\mu=1}^{P}d \mu (z_{\mu}) P(z|\sigma,\xi, \theta) \frac{\partial H_N(\sigma,z|\xi , \theta)}{\partial \tilde{\xi}_i^{\mu}} \\
\nonumber
 &&-  \int_{-\infty}^{+\infty} \prod_{\mu=1}^{P}d \mu (z'_{\mu})  \sum_{\sigma'} P(z',\sigma'|\xi, \theta) \frac{\partial H_N(\sigma',z'|\xi , \theta)}{\partial \tilde{\xi}_i^{\mu}} \Bigg) \\
\nonumber
&=& -\beta \Bigg(  \int_{-\infty}^{+\infty}\prod_{\mu=1}^{P}d \mu (z_{\mu})  \sum_{\sigma} Q(\sigma) P(z|\sigma,\xi, \theta)  \sigma_i z_{\mu}  \\
\nonumber
&& - \sum_{\sigma} Q(\sigma)   \int_{-\infty}^{+\infty}\prod_{\mu=1}^{P}d \mu (z_{\mu}') \sum_{\sigma'} P(z',\sigma'|\xi , \theta) \sigma'_i z'_{\mu}     \Bigg),\\
&=& - \beta \left(  \langle \sigma_i z_{\mu} \rangle_{clamped} - \langle \sigma_i z_{\mu}  \rangle_{free} \right),
\end{eqnarray}
where, in the first passage we used the definition (\ref{KL}), recalling that $Q(\sigma)$ does not depend on $\xi$; in the second passage we used (\ref{eq:ZN5}); in the third passage we used (\ref{eq:ZN1}) and (\ref{eq:ZN2}); in the fourth passage we recalled that $\partial H_N(\sigma,z|\xi)/ \partial  \tilde{\xi}_i^{\mu} = - \sigma_i z_{\mu}$ and the subscript {\em clamped} means that the averages of the two-points correlation functions must be evaluated when the visible layer is forced to assume data values, while {\em free}
means that the averages are the standard, statistical-mechanical ones.}
{For the updating rule of the thresholds $\theta_i (i=1,...,N)$, one performs analogous calculations and, recalling $\partial H_N(\sigma,z|\xi)/ \partial  \theta_i = - \sigma_i$, one gets}
\be
\frac{\partial D(Q,P)}{\partial \theta_i} = - \beta \left(  \langle \sigma_i \rangle_{clamped} - \langle \sigma_i \rangle_{free} \right).
\ee
{Thus, the learning rule (\ref{contrasto}) ultimately} tries to make the theoretical
one-point and two-point correlation functions as close as possible to the empirical ones\footnote{This argument can be expanded up to arbitrarily $N$-points correlation functions by paying the price of adding extra hidden layers and this kind of extension is a basic principle underlying Deep Learning \cite{DL1}.}. {Under this rule the} machine will eventually be able to reproduce the statistics of the training data correctly, and this means that the {parameters $(\xi, \theta)$} have been rearranged such that, if the machine is now asked to generate vectors with $P(\sigma)$, the statistical properties of these vectors will coincide with those of the input data generated by $Q(\sigma)$. In this case we say that the machine {\em has learnt} a representation of the reality it has been fed with. Note that this approach allows a proper statistical reproduction of mean averages and variances, hence, when $Q(\sigma)$ {violates the central limit theorem, a two-layer RBM is no longer suitable for statistical inference}.

\subsection{A brief digression on fast variable's dynamics: retrieval}
{After the learning stage, the machine undergoes a final check over another bulk of data}, referred to as {\em test set}, which stems from the same distribution that has generated the training set \cite{DLbook}. {To fix ideas, let us assume that the machine was trained for retrieval tasks\footnote{{In the neural network jargon with ``retrieval'' we mean that the network is supposed to have learnt a set of data (also called ``patterns''), say a set of pictures, and now it is given, as input, one of those pictures but corrupted or incomplete. If well-performing, the network has to provide, as output, a picture that is as close as possible, to the original one.}}; if the trained machine is able to retrieve correctly the items in the test set, then the test is passed and the machine is ready for the usage}\footnote{The capability to accomplish pattern recognition, hence to
perform as an associative memory, depends on several factors as, e.g., the amount of noise in the data, the amount of data with respect to the amount of neurons to deal with them, etc. \cite{Amit}.}. {In order to move from the learning mode to the retrieval mode, the hidden layer is marginalized over: as we are going to show, following this procedure we end up with a Hopfield model (that is the standard model for pattern retrieval \cite{Hopfield}), where each feature learnt by the hidden layer corresponds to one of the learnt patterns and the optimal parameters $(\xi, \theta)$ store information about the whole set of learnt patterns \cite{Barra0,Agliari1}.}

{To see this duality between the RBM and the Hopfield model we look at the temporal evolution of the neurons which can be described by the following stochastic differential equation and map (to fix ideas we take hidden units as continuous and visible units as binary, as before)}
\begin{eqnarray} \label{eq:dyn}
\frac{d z_{\mu}(t)}{dt} &=& - z_{\mu}(t) + \frac{1}{\sqrt{N}} \sum_{{i=1}}^N \xi_i^{\mu} \sigma_i + \sqrt{\frac{2}{\beta}} \eta_{\mu}(t),\\ \label{s-dyn}
{\sigma}_i(t) &=& \textrm{sign}\left[\tanh \left( \frac{\beta}{\sqrt{N}} \sum_{\mu=1}^P \xi_i^{\mu}z_{\mu} + {\beta \theta}_i \right) + \tilde{\eta}_i {(t)} \right].
\end{eqnarray}
{In the previous equation} we used $t$ to denote the time and we set the typical timescale of the variables $(\sigma, z)$ as unitary; also, we denoted with $\eta, \ \tilde{\eta}$ standard Gaussian
white noises  with zero mean and covariance $\langle \eta_{\mu}(t) \eta_{\nu}(t') \rangle = \delta_{\mu \nu} \delta (t-t')$. {Notice that, in the temporal evolution of the visible (respectively hidden) units, the hidden (respectively visible) units are taken as fixed (see also \cite{Coolen}).\\ Let us now focus on} the hidden layer dynamics:
the first term in the right-hand side of eq.~(\ref{eq:dyn}) is the standard leakage term and the second term is the input signal over the hidden layer. This dynamics overall defines an Ornstein-Uhlembeck process, whose equilibrium distribution, at fixed $\sigma$'s, reads as
\begin{equation}\label{features}
P(z_{\mu}|\sigma) = \sqrt{\frac{\beta}{2\pi}} \exp\left[ -\frac{\beta}{2} \left( z_{\mu} -  \frac{1}{\sqrt{N}} \sum_{i=1}^N \xi_i^{\mu} \sigma_i \right)^2 \right].
\end{equation}
Since the hidden units are independent in the RBMs under study, we can write $P(z|\sigma)=\prod_{\mu=1}^P P(z_{\mu}|\sigma)$.
\newline
{As for the dynamics of} the visible layer, each spin perceives an effective field (that is the sum of the overall signal and the
threshold for firing) that is compared with the noise {in such a way that} if the signal prevails over the noise the neuron spikes. Hence, for the $\sigma$'s, we can write
\begin{equation}
P(\sigma_i|z) = \frac{e^{\frac{\beta}{\sqrt{N}} \sigma_i \sum_{\mu}^P \xi_i^{\mu} z_{\mu} {+ \beta \theta_i \sigma_i} }}{2\cosh\left(\beta \sum_{\mu}^P \xi_i^{\mu} z_{\mu} /\sqrt{N} {+ \beta \theta_i } \right)},
\end{equation}
and, again, $P(\sigma|z) = \prod_{i=1}^N P(\sigma_i|z)$. In order to get
the joint distribution $P(\sigma,z)$ and the marginal distributions $P(\sigma)$, we use Bayes' Theorem, i.e. $P(\sigma, z)=P(\sigma|z)P(z)=P(z|\sigma)P(\sigma)$, hence getting
\begin{eqnarray}\label{pre-hopfield}
P(\sigma,z) &\propto& \exp\left[ -\frac{\beta}{2}\sum_{\mu=1}^P z_{\mu}^2 + \frac{\beta}{\sqrt{N}} \sum_{i=1}^{N}\sum_{\mu=1}^{P} \xi_i^{\mu} \sigma_i z_{\mu} \right],\\ \label{hopfield}
P(\sigma) &\propto& \exp \left[ \frac{\beta}{2N}\sum_{i,j=1}^N \left(\sum_{\mu=1}^P \xi_i^{\mu}\xi_j^{\mu} \right) \sigma_i \sigma_j \right].
\end{eqnarray}
Remarkably, one can see that the features learnt by the machine (see eq. (\ref{features})) {correspond to}
the patterns that the machine will successively be able to retrieve (see eq. (\ref{hopfield})), as this last equation is nothing but the Gibbs probability distribution for the original Hopfield model \cite{Amit,Hopfield}.
\newline
In order to understand this from a statistical-mechanics perspective
it is useful to re-write eq. (\ref{features}) and  eq. (\ref{hopfield}) in terms of the so-called Mattis
magnetization, defined as
\begin{equation}\label{Mattis}
m_{\mu} \dot{=} \frac1N \sum_{i=1}^N \xi_i^{\mu}\sigma_i,
\end{equation}
which represents the overlap between the (visible) spin configuration and the $\mu$-th pattern {$\xi^{\mu}$ (that is, a vector of length $N$)}.
This results in the following two equations
\begin{eqnarray}\label{workingBM}
P(z_{\mu}|m_{\mu}) &=& \sqrt{\frac{\beta}{2\pi}} \exp\left[ -\frac{\beta}{2} \left( z_{\mu} - \sqrt{N}m_{\mu} \right)^2 \right],\\ \label{workingHN}
P(m_{\mu}) &\propto& \exp \left( \frac{N \beta}{2} \sum_{\mu=1}^P m_{\mu}^2 \right).
\end{eqnarray}
Interestingly, equation (\ref{workingBM}) shows that the $z_{\mu}$'s are Gaussian variables centered at the (properly rescaled) value of the related Mattis magnetization,
and can therefore be interpreted as {\em features detectors} because they ``discover'' the
correct Mattis magnetization, given the available patterns.  Furthermore, equation (\ref{workingHN}) {can be interpreted as a Boltzmann factor for the} Hamiltonian $H \propto - \sum_{\mu} m_{\mu}^2$, therefore, a minimum energy argument suffices
to show that, when possible (i.e., in the low-noise limit), the most convenient $\sigma$-states are those such that $m_{\mu} \to \pm 1$ for some $\mu$. In other words,
once these features have been resolved from the data (and automatically stored as patterns in the network) during training, further data will be classified according to these features. An alternative, useful perspective to tackle BM learning can be
understood by looking at eq. (\ref{eq:dyn}): one can check that the signal term in that equation is the rescaled $\mu$-{th} Mattis magnetization, hence if
the visible layer is fed with a pattern to be retrieved (i.e., a vector containing one of the learnt features), then the corresponding Mattis magnetization {rise, hence} supporting the growth of its relative hidden variable $z_{\mu}$, as a staggered magnetic field. On the other hand, if the state $\sigma$ does not coincide with the $\mu$-{th} pattern, $m_{\mu}$ will be zero in the large volume limit, and no net signal will be felt by the corresponding feature detector $z_{\mu}$.
Also, still at this qualitative level, eq. (\ref{workingHN}) shows that the landscape where the $\sigma$'s are allowed to relax
is the same as the landscape of a set of $P$ Curie-Weiss models; recalling that in the Curie-Weiss model the magnetization (in modulus) goes to one in the small-noise limit, here we have that
at least one Mattis magnetization (in modulus) will tend to one in the zero noise limit\footnote{This picture is rather intuitive in the low-load regime (e.g., when $P$ remains finite in the thermodynamic limit), while, in the high-load regime, one should combine a number of Curie-Weiss models that grows linearly with the system size and the landscape exhibits a glassy component.}. To fix ideas, let us say $m_1=1$: from definition (\ref{Mattis}) we realize that this is possible if and only if the $\sigma$'s are parallel (in modulus) to the weights coded in $\xi^1$, and this situation corresponds to the {\em retrieval} of pattern $\xi^1$. Note that, exactly as for the ferromagnets described by the Curie-Weiss model \cite{Amit}, the crucial observable
to detect a retrieval phase (or a ferromagnetic phase) is the model free-energy: by studying this observable (and its derivatives) with respect to the system parameters one can build the phase diagram of the model highlighting regions (and their boundaries) where the emerging behavior of the system is qualitatively different  \cite{Barra0}.
\begin{figure}[tb]
\centering
	\includegraphics[width=8cm]{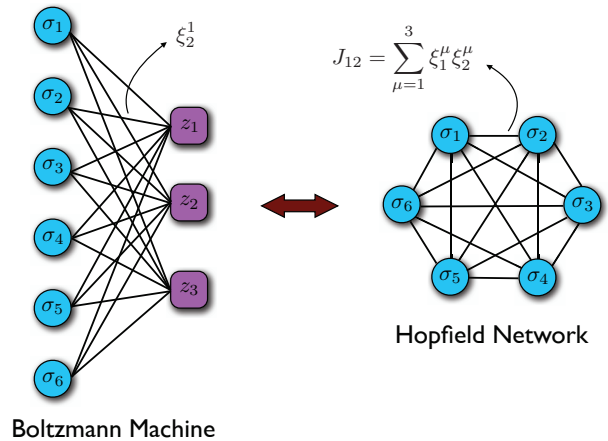}
	\caption{\textbf{Schematic representation of the (restricted) Boltzmann machine (left panel) and its corresponding dual, the Hopfield network.} Left panel: example of a RBM equipped with six neurons in the visible layer, $\sigma_1, ..., \sigma_6$ and three neurons in the hidden layer $z_1, ..., z_3$. The weights among the two layers are coded by the $N \times P$ matrix $\xi_i^{\mu}$. Right panel: dual example of the corresponding Hopfield model, obtained by marginalization over the hidden variables. This network uses solely the $\sigma_1,...,\sigma_6$ neurons, whose links however are now arranged according to the Hebb prescription for learning, that is $J_{ij}=\sum_{\mu=1}^3 \xi_i^{\mu}\xi_j^{\mu}$.}\label{fig:uno}
\end{figure}

We close this section pointing out that for \emph{discrete} weights/patterns the contrastive divergence algorithm shown in the learning section can not be applied
as it requires stochastic descent over the weights that must therefore be \emph{real} (and differentiable). For discrete pattern's entries the most widespread learning rule is the Hebbian {one, which} simply consists in storing patterns {\em adiabatically}, one after another, up to the saturation of the memory: calling $J_{ij}$ the effective coupling between the
two neurons $\sigma_i$ and $\sigma_j$, given $P$ patterns $\xi^{\mu}$, namely $P$
vectors of length $N$ of binary entries, such a prescription results in
\be\label{Hebb}
J_{ij} = \sum_{\mu=1}^P \xi_i^{\mu}\xi_j^{\mu},
\ee
that, formally, does coincide with the {coupling emerging from the duality between RBMs and Hopfield networks, as highlighted by } eq. (\ref{workingHN}).

Before proceeding {with the presentation of theorems for the RBMs}, a couple of remarks are in order. First, for the sake of simplicity, in the following we will omit thresholds $\theta_i (i=1,...,N)$ because their mathematical treatment is trivial as they do not involve correlations
between spins; this does not imply any loss of generality, as thresholds can always be re-introduced at any moment.
Moreover, in the next sections, we will work with random patterns, symmetrically distributed around zero: by a Shannon's compression argument this
is the worst choice and therefore makes our assertions fully general with respect to the choice of patterns. In fact, if the network is able to cope
with $P$ fully random patterns, it will certainly be able to cope with at least $P$ {patterns displaying some degree of correlation (this is the case in real-world applications where correlations are always present, at least, due to the finiteness of $P$, and $N$).}

\section{Restricted Boltzmann machines with Boolean patterns} \label{sec:BT}
 In this section, assuming the existence of the infinite-volume limit for the free energy of RBM with Boolean patterns,
 we will prove the self-average property of the free-energy around its quenched expression and we will give the explicit expression of its annealed approximation.

\subsection{Preliminary Definitions}
\begin{Definition}
We consider a RBM described by the following Hamiltonian:
\begin{equation}\label{BMprima}
H_N(\sigma,z|\xi)=-\frac{1}{\sqrt{N}} \sum_{i=1}^N \sum_{\mu=1}^P \xi^\mu_i z_\mu \sigma_i,
\end{equation}
where the $\sigma_i$ $(i \in [1,...,N])$ are $N$ Ising spins forming the visible layer, the $z_\mu$ $(\mu \in [1,...,P])$ are $P$ Gaussian $\mathcal{N}[0,\beta^{-1}]$
spins forming the hidden layer and the $\xi^\mu_i$ $(i \in [1,...,N], \mu \in [1,..,P])$ are i.i.d. random variables with discrete values $\pm 1$, which
provide the weight of the link connecting the visible unit $i$ and the hidden unit $\mu$.
\end{Definition}
We assume that the two layers scale linearly in their reciprocal volumes, i.e., $\lim_{N\rightarrow \infty} P/N =  \alpha \in \mathbb{R}^+$. This
regime corresponds to the so-called high-load \cite{Coolen}, in contrast with the so-called low-load regime characterized by a vanishing ratio $P/N$ as $N \rightarrow \infty$: the latter is by far less tricky as it does not require the introduction of replicas \cite{Amit,Barra1}, while the former
is the mathematically challenging one \cite{Vieri,Mezard,Rodriguez,Florent1}.
\newline
As anticipated, we are interested in the thermodynamic properties of the neurons (i.e., the {\em fast variables}) in both the parties, while the
weights (i.e., the {\em slow variables}) are assumed to be quenched and drawn randomly from the distribution
\be
P(\xi_i^{\mu}= +1)=P(\xi_i^{\mu}= -1)=1/2,
\ee
independently for all $i= 1,..., N$ and all $\mu = 1, ..., P$.
\begin{Definition}
We introduce the (random) partition function $Z_{P,N}(\beta|\xi)$ as
\begin{equation}\label{laZeta}
Z_{P,N}(\beta|\xi)= \int_{-\infty}^{+\infty} \prod_{\mu=1}^P dz_\mu \sqrt{\frac{\beta}{2 \pi}} e^{-z^2_\mu \beta/2} \sum_{\sigma} e^{-\beta H_N{(\sigma, z | \xi)}},
\end{equation}
and the usual definitions of the intensive free energy $A(\alpha,\beta)$, of the quenched intensive free energy $A^Q(\alpha,\beta)$ and of the annealed intensive free energy $A^A(\alpha,\beta)$ as, respectively,
\begin{eqnarray}
A(\alpha,\beta) &=& \lim_{N \to \infty} A_{P,N}(\beta,\xi), \ \ \  A_{P,N}(\beta, {\xi})  \dot=  \frac{1}{N} \log Z_{P,N}(\beta|\xi),\\
\label{eq:defaq}
A^Q(\alpha,\beta) &=& \lim_{N \to \infty} A^Q_{P,N}(\beta), \ \ \  A_{P,N}^Q(\beta)  \dot=  \frac{1}{N} \mathbb{E} \log Z_{P,N}(\beta|\xi),\\
\label{eq:defaa}
A^A(\alpha,\beta) &=& \lim_{N \to \infty} A^A_{P,N}(\beta), \ \ \  A_{P,N}^A(\beta)  \dot=  \frac{1}{N} \log \mathbb{E} Z_{P,N}(\beta|\xi),
\end{eqnarray}
where the subscripts $P,N$ highlight when we are working at finite volume and the symbol $\mathbb{E}$ represents the average over the quenched variables, that is
\begin{equation}
\mathbb{E}  ~  \dot=\frac{1}{2^{P N}}  \sum_{ \{\xi^\mu_i = \pm 1\}_ {i=1,..,N}^{\mu=1,..,P}}.
\end{equation}
{In the following, if not otherwise specified, the thermodynamic observables are meant as intensive. Also, we recall that the free energy we use is simply a rescaling of the (probably more popular) free energy $f_{P,N} (\alpha,\beta) = -\beta^{-1} A_{P,N}(\alpha,\beta)$.}
\end{Definition}
\begin{Definition}
Given an observable $y(\sigma,z|\xi)$ depending on the neuron configuration and on the weights, once the partition function (\ref{laZeta}) is introduced, we can also define the product Boltzmann state over $s$ replicas of the system as
\be
\nonumber
\Omega(y(\sigma,z|\xi))= \omega_1(y(\sigma^{{1}},z^{{1}}|\xi)) \bigotimes ... \bigotimes \omega_s(y(\sigma^s,z^s|\xi)),
\ee
where {$\omega_i$ represents the Gibbs measure for the $a$-th replica ($a=1,...,s$), that is}
$$
\omega_a(y(\sigma^a,z^a|\xi))= \frac{{\prod_{\mu=1}^P} \int_{-\infty}^{+\infty} dz^a_\mu e^{-(z^a_\mu)^2\beta/2} \sqrt{\frac{\beta}{2 \pi}} \sum_{\sigma^a} y(\sigma^a,z^a|\xi) e^{-\beta H_N {(\sigma^a, z^a | \xi)} }}{Z_{P,N}(\beta|\xi)}.
$$
Finally, we introduce the symbol $\langle y(\sigma,z|\xi) \rangle$ to mean
$$
\langle y(\sigma,z|\xi) \rangle = \mathbb{E}\Omega\left( y(\sigma,z|\xi)\right).
$$
\end{Definition}

\subsection{Main results}

In this subsection  we prove the main Theorem on the self-averaging properties of the free energy $A_{P,N}( \beta, \xi)$ (around its quenched expression, in the thermodynamic limit) of the machine defined by eq. (\ref{BMprima}), then we give the explicit expression for its annealed approximation  $A^A(\alpha, \beta)$.

\begin{Theorem}
The free energy of the Boltzmann machine defined by the Hamiltonian (\ref{BMprima}) is a self-averaging quantity: in the thermodynamic limit its fluctuations vanish and force the latter over its quenched expectation, i.e.
\begin{equation}\label{theorem1}
\lim_{N \to \infty} \mathbb{E}\left [  \big ( A_{P,N}(\beta|\xi) - \mathbb{E} \left(A_{P,N}(\beta|\xi) \right)  \big)^2 \right]=0.
\end{equation}
\end{Theorem}
If the thermodynamic limit of the free energy exists, its quenched expectation is bounded by its annealed one, i.e., $A^Q(\alpha, \beta) = \lim_{N \to \infty} N^{-1}\mathbb{E}\log Z_{P,N}(\beta|\xi) \leq \lim_{N \to \infty} N^{-1}\log \mathbb{E} Z_{P,N}(\beta|\xi) = A^A(\alpha, \beta)$. In fact, for small $\beta$ (that is, $\beta<1$), $A^A(\alpha,\beta)$ coincides with the quenched free energy $A^Q(\alpha,\beta)$, and, in general, it works as its upper bound thanks to Jensen inequality \cite{Guerra1,Guerra2}.
This motivates the next proposition.
\begin{proposition}
The asymptotic value of the annealed free-energy of the RBM is
\begin{equation}
A^A(\alpha,\beta)=\lim_{N \to \infty} \frac{1}{N}\ln \mathbb{E}[ Z_{P,N}(\beta|\xi) ]= \ln 2-\frac{\alpha}{2}\ln(1-\beta).
\end{equation}
\end{proposition}
The proof of Theorem 1 can be found in the Appendix A, while the proof of Proposition 1  is reported hereafter\footnote{It is worth comparing the expression for the free-energy reported in Proposition 1 and the one obtained in \cite{Barra5} for an Hopfield model in the high-load regime. The latter displays an additional term, that is $A^A(\alpha,\beta)=\ln 2-\frac{\alpha}{2}\ln(1-\beta) - \alpha \beta/2$ which stems from the diagonal terms $(P/2)$ in the Hamiltonian $H_N(\sigma|\xi) = - \frac{1}{N} \sum_{i<j} \sum_{\mu} \xi_i^{\mu} \xi_j^{\mu} \sigma_i \sigma_j = - \frac{1}{2N} \sum_{i,j} \sum_{\mu} \xi_i^{\mu} \xi_j^{\mu} \sigma_i \sigma_j + \frac{P}{2}$, which in this work are neglected.}.
\newline
\newline
\textbf{Proof} ~ Recalling the partition function of the model
\begin{equation}
Z_{P,N}(\beta|\xi)= \sum_{\sigma} \prod_{\mu=1}^P  \int_{-\infty}^{+\infty}  d\mu(z_{\mu})  e^{\left( \frac{\beta}{\sqrt{N}}  \sum_i \xi_i^{\mu} \sigma_i \right) z_{\mu} },
\end{equation}
we perform the Gaussian integration to get
\begin{equation} \label{eq:mine}
Z_{P,N}(\beta|\xi)= \sum_{\sigma}  \prod_{\mu=1}^P \sum_{k=0}^{\infty}
\left ( \frac{\beta}{2 N}\right)^k \frac{\mathcal{K}_{\mu}(\sigma| \xi)^{2k}}{k!} ,
\end{equation}
where, for brevity, we posed $\mathcal{K}_{\mu}(\sigma| \xi) = \sum_{i=1}^N \xi_i^{\mu} \sigma_i$. We now average over the weights $\{ \xi_i^{\mu}\}$ to get $\mathbb{E}[Z_{P,N}(\beta|\xi)]$ and, to this goal, we focus on $\mathbb{E} [ \mathcal{K}_{\mu}(\sigma| \xi)^{2k}]$:
\begin{eqnarray}
\nonumber
\mathbb{E} [\mathcal{K}_{\mu}(\sigma| \xi)^{0}  ]&=& 1,\\
\nonumber
\mathbb{E} [\mathcal{K}_{\mu}(\sigma| \xi)^{2} ] &=& \sum_{i,j=1}^N \sigma_i \sigma_j  \mathbb{E} [\xi_i^{\mu}\xi_j^{\mu}] = \sum_{i,j=1}^N \sigma_i \sigma_j \delta_{ij} = N,\\
\nonumber
\mathbb{E}[\mathcal{K}_{\mu}(\sigma| \xi)^{4} ] &=& \sum_{i,j,k,l=1}^N \sigma_i \sigma_j \sigma_k \sigma_l  \mathbb{E}[ \xi_i^{\mu}\xi_j^{\mu} \xi_k^{\mu}\xi_l^{\mu} ],  \\
\nonumber
&=& \sum_{i,j,k,l=1}^N [(\delta_{ij}\delta_{kl} + \delta_{ik}\delta_{jl}  + \delta_{il}\delta_{jk}) + \delta_{ijkl} -3 \delta_{ijkl}] \sigma_i \sigma_j \sigma_{k} \sigma_l = 3N^2 - 2N,
\end{eqnarray}
where we exploited the orthogonality of the weights and, in particular, in the last line the non-null contributions correspond to pair-wise equal indices and four-wise equal indices, properly accounting for repetitions.
This calculation can be generalized to higher order, yet it gets more awkward and, ultimately, not necessary to our purposes. In fact, one can notice that, as $N \rightarrow \infty$, the leading term for $\mathbb{E} [\mathcal{K}_{\mu}(\sigma| \xi)^{2k} ]$ is order $N^k$ while subleading terms get vanishing under logarithm and normalization, as prescribed by the formula (\ref{eq:defaa}). Such leading term is given by $\prod_{i=0}^{k-1}(2k-2i-1)N^k$, which accounts for only pair-wise contributions, and
\begin{eqnarray} \label{eq:b1}
   \left [ 1 +\sum_{k=1}^{\infty}  \left ( \frac{\beta}{2N}\right)^k \frac{\prod_{i=0}^{k-1}(2k-2i-1) N^k}{k!} \right] = \frac{1}{\sqrt{1-\beta}}.
\end{eqnarray}
Consequently, stressing that the average over $\xi$ allows getting rid of the $\sigma$ variables, so that the sum over all $\sigma$ configurations simply provides a factor $2^N$, and the dependence on $\mu$ is dropped so the product over $\mu$ simply provides a power $P$, we can write
\begin{equation}
\lim_{N \rightarrow \infty} \frac{1}{N} \ln \mathbb{E}[ Z_{P,N}(\beta|\xi)] =  \lim_{N \rightarrow \infty}  \frac{1}{N}\ln \left[  2^N (1 - \beta)^{-P/2} \right] = \ln 2 -  \frac{\alpha}{2} \ln(1 - \beta).
\end{equation}

\section{Restricted Boltzmann machines with Gaussian patterns.} \label{sec:FG}
In this section, assuming the existence of the infinite-volume limit for the free energy of the RBM with Gaussian patterns, we will provide an explicit expression for the quenched free energy, under the assumption of replica symmetry.
\newline
We remark that in this entire section we will cover solely the technical aspects from a mathematical perspective, while we leave the next section for technical aspects from an applicative perspective.

\subsection{Preliminary Definitions}
\begin{Definition}
We consider a RBM described by the following Hamiltonian
\begin{equation}\label{BMseconda}
H_N(\sigma, z |\xi)=-\frac{1}{\sqrt{N}} \sum_{i=1}^N \sum_{\mu=1}^P \xi^\mu_i  z_\mu \sigma_i,
\end{equation}
where $\sigma_i$ $(i \in [1,...,N])$ are $N$ Ising spins forming the visible layer, $z_\mu$ $(\mu \in [1,...,P])$ are $P$ Ising spins forming the hidden layer, while $\xi^\mu_i$ $(i \in [1,...,N], \mu \in [1,..,P])$ are quenched and, in particular, this time are chosen as i.i.d. Gaussian random variables $\mathcal{N}[0,1]$, namely the weights connecting the two layers are analogical \cite{Barra3,Barra5}.
\end{Definition}
Again, we assume the two layers to scale linearly in the reciprocal volume, i.e., $\lim_{N \rightarrow \infty} P/N = \alpha  \in \mathbb{R}^+$.
The definitions of the partition function and the free energy (as well as its quenched and annealed expectations) coupled to the Hamiltonian (\ref{BMseconda}) are analogous to those given in the previous section, apart for the average over the patterns: the symbol $\mathbb{E}$ now represents the average over all the Gaussian distributed (and no longer Boolean) pattern entries, namely
\be
\mathbb{E} \dot{=}  \int_{-\infty}^{+\infty}\prod_{i=1}^{N} \prod_{\mu=1}^{P} d \xi_i^{\mu} \frac{e^{- (\xi_{i}^{\mu})^2/2}}{\sqrt{2 \pi}}.
\ee
\begin{Definition}
The (spin-glass) order parameters of the RBM defined by eq. (\ref{BMseconda}) are the {two-replica} overlap in the visible layer $q_{12}(\sigma)$
and the two-replica overlap in the hidden layer $p_{12}(z)$ defined as
$$
q_{12} \dot= \frac{1}{N}\sum_{i=1}^N \sigma_i^1 \sigma_i^2, \ \ \ \ \ \ \ \   p_{12} \dot= \frac{1}{P}\sum_{\mu=1}^P z_{\mu}^1 z_{\mu}^2,
$$
\end{Definition}

\subsection{Preliminary results}

In this subsection we present our strategy of investigation; namely we prove some
theorems and propositions whose implications will be exploited in the following subsection. Taken a real scalar $n \in (0,1]$
\be\label{ReplicaTrick}
\mathbb{E}\ln Z_{P,N}(\beta|\xi) = \lim_{n\to 0}\frac{1}{n}\ln\mathbb{E}\left[Z^{n}(\beta|\xi)\right],
\ee
and we use this relation in order to write $A^Q(\alpha,\beta)=\lim_{N\to\infty}\lim_{n \to 0} \phi_{P,N}(\beta,n)$, where
\begin{definition}
The $n$-quenched free energy $\phi_{P,N}(\beta,n)$ has the following interpolating functional
\be \label{eq:defphi}
\phi_{P,N}(\beta,n) = \frac{1}{Nn} \ln \mathbb{E} \left[ Z_{P,N}^n(\beta|\xi)\right].
\ee
\end{definition}
The scalar $n$ is thought of as an interpolating parameter such that $\phi_{P,N}(\beta,n)$ recovers the annealed free-energy when $n=1$ and the quenched free-energy when $n \rightarrow 0$. The former statement can be checked by directly replacing $n=1$ in (\ref{eq:defphi}) and comparing with (\ref{eq:defaq}); the latter statement is given by the following
\begin{Theorem}\label{ThNew1}
The following relations between the $n$-quenched free energy and the quenched free energy exist
\be \label{eq:th2}
\lim_{n\to 0} \phi_{P,N}(\beta,n) = A^Q_{P,N}(\beta),
\ee
furthermore
\be \label{eq:secthe}
\phi_{P,N}(\beta,n) \geq A^Q_{P,N}(\beta),
\ee
for any $n$.
\end{Theorem}
%
\begin{proof}
We can Taylor-expand the $n$-quenched free energy for $n$ close to $0$ to get
\begin{eqnarray}
\nonumber
\ln \mathbb{E}\left[Z_{P,N}^n(\beta|\xi) \right] &=& 0 + \frac{\mathbb{E} [Z^n_{P,N}(\beta|\xi) \log Z^n_{P,N}(\beta|\xi)] }{\mathbb{E} Z^n_{P,N}(\beta|\xi)} n  + o(n^2) \\
\nonumber
&=& \mathbb{E}\left[ \ln Z_{P,N}(\beta|\xi) \right] \cdot n + o(n^2) \Rightarrow \\
\lim_{n \to 0} \phi_{P,N}(\beta,n) &=& \lim_{n \to 0} \frac{1}{nN}\left[\mathbb{E}\left( \ln Z_{P,N}(\beta|\xi) \cdot n + o(n^2) \right) \right] = A^Q_{P,N}(\beta),
\end{eqnarray}
which proves (\ref{eq:th2}), while (\ref{eq:secthe}) is guaranteed by the Jensen inequality.
\end{proof}

\subsection{Main results}
Our strategy is based on two interchained interpolation schemes: one, anticipated before, that works on replicas, the other, described hereafter, that works on spin coupling.
Before proceeding, we stress that we will exploit our strategy solely within the so-called {\em replica symmetric approximation}, namely under the assumption that, in the asymptotic limit, the order parameters do not fluctuate (i.e., their distributions get delta-peaked over their thermodynamic averages).
\newline
Let us now introduce the interpolating structure acting on couplings. First, we give a few definitions which will lighten the calculations along the way.
\begin{definition}
We introduce, as an interpolating parameter, a real scalar $t \in [0,1]$, further we need $N$ i.i.d. $\mathcal{N}[0,1]$ random variables, referred to as $J_i, \ \ \ i \in (1,...,N)$, and $P$ i.i.d. $\mathcal{N}[0,1]$ random variables, referred to as $J_{\mu},\ \ \ \mu \in (1,...,P)$, and we define the interpolating partition function $Z_t$ as
\be\label{Z-interpolante}
Z_t = \sum_{\sigma} \int_{-\infty}^{+\infty}\prod_{\mu=1}^{P}d \mu(z_{\mu}) e^{\sqrt{t}\left( \frac{\sqrt{\beta}}{\sqrt{N}}\sum_{i,\mu}^{N,P}\xi_i^{\mu}\sigma_i z_{\mu} \right) + \sqrt{1-t}\left( \sqrt{\alpha \beta \bar{p}} \sum_{i}^{N} J_i \sigma_i + \sqrt{ \beta \bar{q}} \sum_{\mu}^{P} J_{\mu}z_{\mu}\right) + (1-t)\beta (1-\bar{q}) \sum_{\mu}^{P} \frac{z_{\mu}^2}{2}}.
\ee
Notice that when $t=1$ we recover the partition function $Z_{P,N}(\beta|\xi)$ of the RBM under study, while when $t=0$ we get the partition function of a one-body model for the variables $\sigma$ and $z$.
We also remark that now the quenched expectation $\mathbb{E}$ applies also over $(J_i, J_{\mu})$, i.e., $\mathbb{E}=\prod_{i,\mu}\mathbb{E}_{\xi_i^{\mu}}\prod_{i,\mu}\mathbb{E}_{J_i}\mathbb{E}_{J_{\mu}}$.
\end{definition}
\begin{definition}
Given a smooth function $f(\sigma,z|\xi)$ of the neurons $(\sigma,z)$ and of the patterns $\xi$, we introduce a deformed measure $\langle f(\sigma,z|\xi) \rangle_n$ as
\be
\langle f(\sigma,z|\xi) \rangle_n \equiv \mathbb{E}\left[Z_t^n \cdot (\mathbb{E}\left(Z_t^n\right))^{-1} \cdot \Omega\left( f(\sigma,z|\xi)\right)  \right],
\ee
where $\Omega$ is the standard generalized $2$-product Boltzmann state, namely $\Omega = \omega_1 \bigotimes \omega_2$.
\newline
Note that this new measure collapses on the standard one for $n\to 0$, namely $\langle f(\sigma,z|\xi) \rangle_n \to \langle f(\sigma,z|\xi) \rangle$ as $n \to 0$.
\end{definition}
\begin{definition}
Once introduced a  scalar parameter $t \in [0,1]$, the following interpolating functional
\be  \label{eq:phibetant}
\phi_{P,N}(\beta, n,t) = \frac{1}{Nn} \ln \mathbb{E}\left( Z_t^n \right)
\ee
bridges the $n$-quenched free energy of the RBM to an integrable one-body model.
\end{definition}
With these definitions we can state
\begin{proposition}
The following boundary values for the interpolating $n$-quenched free energy hold
\begin{eqnarray}
\phi_{P,N}(\beta, n,t=1) &=& \frac{1}{Nn} \ln \mathbb{E}\left( Z_1^n \right) = \phi_{P,N}(\beta,n), \\ \nonumber
\phi_{P,N}(\beta, n,t=0) &=& \ln 2 + \frac{1}{n}\ln \int_{-\infty}^{+\infty}d\mu(x)\cosh^n \left(x \sqrt{\alpha \beta \bar{p} } \right) - \frac{\alpha}{2}\ln\left[ 1-\beta \left(1-\bar{q}\right) \right]\\ &+& \frac{\alpha \beta}{2} \frac{\bar{q}}{1-\beta\left(1-\bar{q}\right)}.
\end{eqnarray}%
\end{proposition}
\begin{Theorem}
In the replica symmetric approximation, the thermodynamic limit of the free energy $A_{P,N}(\beta,\xi)$ converges to its quenched expectation, that reads as
\begin{eqnarray}
\nonumber
A^Q(\alpha,\beta) &=& \ln 2 +   \int_{-\infty}^{+\infty}d\mu(x) \ln\cosh  \left(x  \sqrt{\alpha \beta  \bar{p} } \right)
 -\frac{\alpha}{2}\ln \left[1-\beta(1-\bar{q}) \right]\\
 \label{RS-solution}
 &+& \frac{\alpha \beta}{2} \frac{\bar{q}}{1-\beta(1-\bar{q}) } -\frac{\alpha \beta}{2} \bar{p} \left(1 -  \bar{q} \right).
\end{eqnarray}
whose extremization with respect to $\bar{q},\ \bar{p}$ returns
\begin{eqnarray}\label{self-x-pbarra}
\frac{\partial A^Q(\alpha,\beta)}{\partial \bar{q}} &=& 0 \Rightarrow \bar{p} = \frac{ \beta \bar{q}}{\left[1-\beta \left(1-\bar{q}\right)\right]^2}, \\
\frac{\partial A^Q(\alpha,\beta)}{\partial \bar{p}} &=& 0 \Rightarrow \bar{q} = \int d \mu(x) \tanh^2\left(x  \sqrt{\alpha \beta \bar{p}}  \right).
\end{eqnarray}
\end{Theorem}
\begin{corollario}
The RBM has an ergodic regime (where the annealed approximation holds and $A^A(\alpha,\beta)$ equals the quenched replica symmetric free energy $A^Q(\alpha,\beta)$ evaluated at $\bar{q}=\bar{p}=0$) that is bounded by the following transition line, in the $\alpha,\beta$ plane of tuneable parameters
\be \label{eq:criticalline}
 \beta_c = \lim_{n \to 0} \frac{1}{1+\sqrt{\alpha/(1-n)}} = \frac{1}{1+\sqrt{\alpha}}.
\ee
\end{corollario}
The proofs of Proposition 2, Theorem 3 and Corollary 1 can be found in the Appendix B.
\begin{remark}
We note that the critical line delimiting the ergodic region for the RBMs is the same critical line that traces the ergodic boundary for the Hopfield model \cite{Amit,Tirozzi2,Barra5}. However, while here $\alpha$ is the ratio between the hidden and visible layer sizes, in the Hopfield model $\alpha$ is the ratio between the stored patterns and the neurons necessary to handle them.
\newline
Further, the annealed expression for the free energy of this RBM is the same as the one obtained for the model treated in Sec. 3 in the high noise limit (see Proposition $1$). In fact, by setting $\bar{q}=\bar{p}=0$ in eq. (\ref{RS-solution}) we get
\be
A^A(\alpha,\beta) =\ln 2  -\frac{\alpha}{2}\ln \left(1-\beta\right).
\ee
\end{remark}

\section{The signal and the noise} \label{sec:signal}
In this section we deepen the inferential capabilities of these machines and the properties of the inner (slow) noise. More precisely, we introduce a ``signal'', namely an external field which favors the retrieval of a given pattern and, by inspecting its corresponding Mattis magnetization, we check whether and how robustly (with respect to the parameters $(\alpha, \beta$) this pattern can be detected over the noise. Next, we further address the noise and we provide
an argument for its universality to hold, namely we will show that, in the thermodynamic limit, it does not matter if we consider Boolean or Gaussian patterns, they contribute in the same way to the (slow) noise.
\begin{proposition}
We consider one Boolean pattern $\xi^1$ and $P-1$ Gaussian patterns $\xi^{\mu}$, $\mu=2,...,P$; given a signal which favours the retrieval of pattern $\xi^1$, we measure the accuracy of the signal reconstruction in terms of the related Mattis magnetization $m_1(\sigma)$ which fulfils an Hopfield-like self-consistent equation.
\end{proposition}
\begin{proof}
The proof works by extending the interpolating scheme coded by eq. (\ref{Z-interpolante}) in order to account for a signal term: in this way we get a self-consistent equation for the related Mattis magnetization $\bar{m}_1$.
\newline
Let us consider the interpolating partition function (\ref{Z-interpolante}): it is enough to add to its exponent the terms $(1-t)  \beta \bar{m}_1 \sum_i \xi_i^1 \sigma_i + t  \beta m_1^2(\sigma) $ to extend the previous quenched free energy, still at the replica symmetric level (i.e., assuming $\lim_{N \to \infty}P(m_1)=\delta(m_1-\bar{m}_1)$ to obtain
\begin{eqnarray}\label{RS-solution2}
A^Q(\alpha,\beta) &=& \ln 2 +   \int_{-\infty}^{+\infty}d\mu(x) \ln\cosh  \left(x  \sqrt{\alpha \beta \bar{p}} + \beta \bar{m}_1  \right)
 -\frac{\alpha}{2}\ln \left[1-\beta(1-\bar{q}) \right] \\ &+& \frac{\alpha \beta}{2} \frac{\bar{q}}{1-\beta(1-\bar{q})} -\frac{\alpha \beta}{2} \bar{p} \left(1 -  \bar{q} \right) - \frac{\beta}{2} \bar{m}_1^2.
\end{eqnarray}
whose extremization w.r.t. $\bar{m}_1$ returns a self consistency equation for the signal that reads
\be \label{eq:scagain}
\frac{\partial A^Q(\alpha,\beta)}{\partial \bar{m}_1} = 0 \Rightarrow \bar{m}_1 = \int_{-\infty}^{+\infty}d\mu(x) \tanh\left(x \sqrt{\alpha \beta  \bar{p}} + \beta \bar{m}_1 \right). \\
\ee
By replacing $\bar{p}$ with the expression in terms of $\bar{q}$ given by (\ref{self-x-pbarra}), eq. (\ref{eq:scagain}) recovers the self-consistent equation for the Mattis magnetization in the Hopfield scenario as traced by the AGS theory \cite{Amit,Tirozzi1}.
\end{proof}
\newline
\newline
This result is naturally consistent with the intrinsic analogies between the Hopfield model and the RBM, as depicted by eqs. (\ref{pre-hopfield},\ref{hopfield}).
\newline
We remark further that the slow noise enters and plays a crucial role in the self-consistency for the signal, thus this should be properly estimated in machine learning applications. It is therefore important to state the next
\begin{proposition}
The noise generated by the not-retrieved patterns over the signal is universal, namely, in the thermodynamic limit, it is the same if the pattern entries are either sampled from a Boolean distribution or from a Gaussian distribution.
\end{proposition}
\begin{proof}
Let us consider the $n$-product of the partition function, taking $\xi^1$ as the retrieved pattern (whose signal is detected by $m_1(\sigma))$ we perform the quenched average over the $P-1$ quenched patterns (overall acting as a slow-noise against the retrieval of $\xi^1$) in complete generality as
\begin{eqnarray}
\mathbb{E}Z_{P,N}^n(\beta|\xi) &=& \mathbb{E} \prod_{a=1}^n \prod_{i=1}^N \sum_{\sigma_i^a} \int_{-\infty}^{+\infty} \prod_{a=1}^n \prod_{\mu=1}^{P}d\mu(z^a_{\mu})e^{\left( \frac{\beta N}{2} \sum_{a}^{n} m_1^2(\sigma^{a})+\sqrt{\frac{\beta}{N}}\sum_{a}^n \sum_{i}^{N}\sum_{\mu}^{P}\xi_i^{\mu}\sigma_i^{a} z_{\mu}^{a} \right)} \\
   &=&  \prod_{a=1}^n \prod_{i=1}^N \sum_{\sigma_i^a}  \int_{-\infty}^{+\infty} \prod_{a=1}^n \prod_{\mu=1}^{P}d\mu(z_{\mu}^a)e^{\left( \frac{\beta N}{2} \sum_{a}^{n} m_1^2(\sigma^{a})\right)}\mathbb{E}e^{\left(\sqrt{\frac{\beta}{N}}\sum_{a}^n \sum_{i}^{N}\sum_{\mu}^{P}\xi_i^{\mu}\sigma_i^{a} z_{\mu}^{a} \right)}
\end{eqnarray}
As the distributions share unitary variance and are both symmetric, the quenched average over the patterns returns
\begin{eqnarray}
\mathbb{E}\left[\exp \left(\sqrt{\frac{\beta}{N}}\sum_{a=1}^n \sum_{i=1}^{N}\sum_{\mu=1}^{P}\xi_i^{\mu}\sigma_i^a z_{\mu}^a \right)\right]
\sim \exp\left(\frac{\beta}{2 N} \sum_{i=1}^{N}\sum_{\mu=1}^{P} \sum_{a,b=1}^{n}\sigma_i^a \sigma_i^b z_{\mu}^a z_{\mu}^b \right)
\end{eqnarray}
\end{proof}
\newline
We remark that, while differently sampled patterns behave as the same source of noise, actually they generate quite different signals \cite{Barra0}, i.e. universality holds just for the noise.
\newline
An interesting feature of these networks, shared by several other glassy systems \cite{Aizenman,GG}, is that their overlaps (i.e., $q_{12}$ and $p_{12}$) display complex, non-Gaussian fluctuations (and their complexity explains why usually researchers work at the replica symmetric level). This is well known in the spin-glass Literature, where the linear constraints fulfilled by fluctuations are usually named Aizenman-Contucci polynomials \cite{Aizenman} (while the more general ones are known as Ghirlanda-Guerra identities \cite{GG}). Nonetheless, this feature is largely ignored among researchers in Machine Learning and hereafter we show that these constraints actually hold also for these networks. In particular, due to the theory developed up to this point, we will focus in details only on the former, namely on the Aizenman-Contucci identities (suitably adapted to the present case):
\begin{proposition}\label{ACcorollary}
In the thermodynamic limit, $\beta$ almost-everywhere, the following Aizenman-Contucci polynomials hold for the RBM's order parameters $q_{12},\ p_{12}$:
\be\label{ACprima}
 \langle q_{12}^2p_{12}^2 \rangle - 4 \langle
q_{12}p_{12}q_{23}p_{23} \rangle + 3 \langle
q_{12}p_{12}q_{34}p_{34}\rangle  =0.\ee
\end{proposition}
\begin{proof}
It is enough to explicitly write down the $\beta$ streaming of their expectation, namely
\begin{eqnarray}
\partial_{\beta} \langle q_{12}p_{12}\rangle &=&
\frac{1}{N P}\sum_{i,\mu} \mathbb{E} \partial_{\beta} \Omega^2
(z_{\mu}\sigma_i) = \frac{1}{N P}\sum_{i,\mu} \mathbb{E} 2 \Omega
(z_{\mu}\sigma_i)\partial_{\beta} \Omega(z_{\mu}\sigma_i) \\ &=&
\frac{2}{N P}\sum_{i,j,\mu, \nu}\mathbb{E}\Omega(z_{\mu}\sigma_i)\xi_j^{\nu}\left[
\Omega(z_{\mu}\sigma_i z_{\nu}\sigma_j)
-\Omega(z_{\mu}\sigma_i)\Omega(z_{\nu}\sigma_j) \right],
\end{eqnarray}
and we use Wick's theorem on $\xi$ to get
\begin{eqnarray}\nonumber
\partial_{\beta} \langle q_{12}p_{12}\rangle &=&
\frac{2}{N^2 P^2}\sum_{\mu,\nu,i,j}\Big \{ \Big [
\Omega(z_{\mu}\sigma_i z_{\nu} \sigma_j)-\Omega(z_{\mu}\sigma_i)\Omega(z_{\nu}\sigma_j)
\Big] \Big [ \Omega(z_{\mu}\sigma_i z_{\nu}\sigma_j)-
\\ \nonumber
&+& \Omega(z_{\mu}\sigma_i)\Omega(z_{\nu}\sigma_j) \Big] +
\Omega(z_{\mu}\sigma_i)\Big[ \Omega(z_{\mu}\sigma_i\sigma_j z_{\nu} z_{\nu}\sigma_j)
- \Omega(z_{\mu}\sigma_i z_{\nu}\sigma_j)\Omega(z_{\nu}\sigma_j)
\\ \nonumber
&-&
\Omega(z_{\mu}\sigma_i)\Omega(z_{\mu}\sigma_i z_{\nu}\sigma_j)\Omega(z_{\nu}\sigma_j)
+
\Omega(z_{\mu}\sigma_i)\Omega(z_{\nu}\sigma_j)\Omega(z_{\nu}\sigma_j)\Omega(z_{\mu}\sigma_i)
\\ \nonumber
&-&
\Omega(z_{\mu}\sigma_i)\Omega(z_{\mu}\sigma_i)\Omega(z_{\nu}\sigma_j)
\Omega (z_{\nu}\sigma_j) +
\Omega(z_{\mu}\sigma_i)\Omega(z_{\nu}\sigma_j)\Omega(z_{\nu}\sigma_j)\Omega(z_{\mu}\sigma_i) \Big ]
\Big \}.
\end{eqnarray}
Finally, by re-introducing the overlaps and making all the cancellations that happen to appear, we can write
\be
\partial_{\beta} \langle q_{12}p_{12}\rangle = \alpha N \Big( \langle q_{12}^2p_{12}^2 \rangle -4  \langle q_{12}p_{12}q_{23}p_{23} \rangle + 3 \langle q_{12}p_{12}q_{34}p_{34} \rangle \Big)
\ee and, by a stability requirement (i.e., this streaming is bounded), again in the thermodynamic limit  the thesis is proved.
\end{proof}

\section{Conclusions}
In this work we addressed RBMs from a mathematical physics perspective.
\newline
First, we provided a basic background on these systems trying to explain, in a non-rigorous way, the mechanisms through which they are able to perform statistical learning and information retrieval; these capabilities make RBMs key tools in the rapidly expanding field of Artificial Intelligence.
In particular, we showed how to obtain the contrastive divergence algorithm for learning of continuous weights and how, ultimately, this recipe displays the same mathematical representation of the Hebb rule, a well consolidated prescription for learning with discrete weights \cite{Hopfield}. Further, {by a robust Bayesian argument, we showed a formal equivalence between RBMs and Hopfield networks, where} the weights learnt by the RBM are exactly the patterns successively retrieved by the (dual) Hopfield network \cite{Barra2}. {Although learning and retrieval are typically addressed separately (thanks to adiabatic-like hypothesis), these results restore the intrinsic interplay between {\em learning} and {\em retrieval} which, in fact, are two aspects of the single and broader act of {\em cognition}.}

Then, adopting a mathematical rigorous treatment, we proved several results concerning the two main representations of RBMs
(analogical versus digital weights). From the statistical mechanical perspective, these systems are two-party spin-glasses whose external party (or {\em layer} to keep the original jargon) is always built up by binary Ising spins/neurons, that are fed by the external information; the inner layer, instead, detects features contained in the data presented as input by suitable learning algorithms which tune the weights connecting the two layers  (e.g., contrastive divergence and simulated annealing).
\newline
Once the learning stage is over, the system relaxes to a Gibbs-measure generated by the Hamiltonian of the RBM, whose control is entirely
deducible by the knowledge of its free energy: the latter thus plays as the principal observable whose {investigation is in order}. However, while
these systems are becoming pervasive in the applied world, at the rigorous level, solely the Sherrington-Kirkpatrick is a fully understood spin-glass (the single-party \cite{MPV,Panchenko}) and any variation on theme (even possibly mild) is still a great challenge and much care should be paid in their mathematical treatment.

Here we showed that the free energy of these machines is a self-averaging quantity in the thermodynamic limit and that it self-averages over its quenched expectation. This was proven by adapting the Pastur-Shcherbina-Tirozzi argument \cite{Tirozzi1,Tirozzi2} (based on Doob Theorem for martingales), originally developed for standard Hopfield networks \cite{Tirozzi1,Tirozzi2}. Then, en route to the quenched free energy, we gave an expression of its annealed approximation that, due to Jensen inequality, plays as a natural bound for the quenched one.\\
Next, turning to machines with Gaussian weights, by an adaptation of Guerra's interpolation scheme \cite{Mingione}, we provided the explicit expression of the free energy of these machines in the thermodynamic limit, under the assumption of replica symmetry.  Remarkably, such a free energy turns out to share the same mathematical structure of the Hopfield one. Furthermore, we gave an argument based on the universality of the quenched noise (consistent with \cite{Genovese}) to generalize such a result to the previous machine with Boolean weights. The ergodicity line, that separates the region where the annealed approximation holds from the region where the quenched one prevails, turns out to be the same of the AGS theory for the Hopfield model \cite{Amit,Tirozzi2} thus conferring a bulk of robustness to the theory as a whole. Finally, we gave a glance at overlaps fluctuations, i.e., beyond the replica symmetry, and we found Aizenman-Contucci identities.

\section*{Acknowledgments}
\noindent
E.A. acknowledges financial support by Sapienza Universit\`a di Roma (project no. RG11715C7CC31E3D).\\
A.B. acknowledges financial support by Salento University, by INFN, by MIUR (via the basic research grant 2018) and by ``Match/Pytagoras" project (FESR/FSE 2014/2020).\\
E.A. and A.B. acknowledge financial support by GNFM-INdAM (``Progetto Giovani 2018'').
\newline
The authors are grateful to Francesco Guerra and Emanuele Mingione for usueful conversations.

\section*{Appendix A. Proof of Theorem 1} \label{sec:A}
\begin{proof}
First, let us explain the idea behind this proof.
We consider the set of all possible $2^{NP}$ weight combinations $\mathcal{R}_N = \{\xi^\mu_i  \}_{i =1,...,N}^{\mu=1,...,P}$ and we build a sequence of subsets
\begin{equation} \label{eq:subset}
\mathcal{R}_k^N=\{\xi^\mu_i  \}_{i \ge k}^{\mu=1,...,P},
\end{equation}
with $k=1,...,N$. 
Next, we define the conditional expectations $\mathbb{E}_{<k}$ over these subsets and we show that the conditional expectation $F_N^k$ of the extensive free energy over two consecutive subsets can be used to estimate the fluctuations in the intensive free energy. Finally, exploiting the martingale nature of the
sequence $(F_N^k, \mathcal{R}_k^N)$, we can bound fluctuations and show that they are vanishing in the thermodynamic limit.
\newline
This method has been widely exploited in the past (see e.g., \cite{Tirozzi2,Tirozzi3,Tirozzi4}) as it can be applied to different examples of neural networks,
provided that patterns are independent and the probability of the conditioning subsets are non null (we refer to \cite{TirozziBook} for a more extensive
discussion on this method).

Let us now describe the proof in more details. The sequence of subsets $\mathcal{R}_k^N$ with $k=1,...,N$, defined in (\ref{eq:subset}), constitutes a decreasing family of $\sigma$-algebras, over which we define
the conditional expectation
\begin{equation}\label{eq:conditional}
\mathbb{E}_{<k} \dot=  { \frac{1}{2^{P(k-1)}}} \sum_{\{\xi^\mu_i = \pm 1 \}_ {i<k}^{\mu=1,..,P}}.
\end{equation}
In particular, the conditional expectation over $\mathcal{R}_k^N$ of the extensive free energy $F_N$ is referred to as $F^k_N$ and reads as
\begin{equation} \label{eq:fk}
F^k_N=\mathbb{E}(F_N | \mathcal{R}_k^N)= {\frac{1}{2^{P(k-1)}}} \sum_{(\xi^\mu_i= \pm 1)_ {i<k}^{\mu=1,..,P}} \left[ \log Z_{P,N}(\beta|\xi) \right].
\end{equation}
Notice that in the left hand side the conditioning is with respect to any event of the $\sigma$-algebra $\mathcal{R}_k^N$ (namely with respect to any choice of
the random variables $\xi_i^{\mu}$, with $i \geq k$) and in the right hand side
the average is performed over the non-conditioned events (namely over all possible weigths $\xi_i^{\mu}$, with $i < k$).
Otherwise stated, the expectation $\mathbb{E}( \cdot | \mathcal{R}_k^N)$ provides the average over a subset of the quenched variables, corresponding to $i<k$,
while any particular realization of weights in $\mathcal{R}_k^N$ can be taken as a ``boundary condition''.
Remarkably, the sequence constituted by the stochastic variables $F_N^k$ and by the $\sigma$-algebra $\mathcal{R}_k^N$ fulfill the martingale
property \cite{Tirozzi1,Tirozzi2}
\be \label{eq:martin}
\mathbb{E}(F^k_N|\mathcal{R}_l^N)=\begin{cases}F^k_N &{if}\quad l<k \\ F^l_N &{if}\quad l \ge k \end{cases}.
\ee

Now, by setting $k=1$ and $k=N+1$ in Eq.~\ref{eq:fk}, we get
\begin{equation}
F^1_N=  \log Z_{P,N}(\beta|\xi), \quad F^{N+1}_N=  \mathbb{E}\left(\log Z_{P,N}(\beta|\xi)\right),
\end{equation}
namely, we recover the extensive free-energy and quenched extensive free-energy, respectively.
Let also introduce the relative increment $\Psi_k$ as
\be \label{eq:psi}
\Psi_k=F^k_N-F^{k+1}_N,
\ee
in such a way that we can write the linear term inside eq. (\ref{theorem1}) as an arithmetic average of the incremental free-energy $\Psi_k$ over the $N$ visible
neurons, that is
\be \label{eq:flutt1}
A_{P,N} - \mathbb{E}\left(A_{P,N}\right)= \frac{1}{N} \sum_{k=1}^N \Psi_k.
\ee
The free-energy fluctuations therefore reads as
\begin{equation} \label{eq:flutt2}
\mathbb{E}\left( A_{P,N}- \mathbb{E}\left(A_{P,N}\right)\right)^2=  \frac{1}{N^2} \sum_{k=1}^N \mathbb{E} \Psi_k^2 +\frac{2}{N^2} \sum_{k=1}^N \sum_{l=k+1}^N \mathbb{E} \Psi_k \Psi_l,
\end{equation}
where in the right hand side we split the diagonal and the off-diagonal contributions.
Actually, only the former matters since the latter is null as shown in the following:
\begin{equation} \label{eq:flutt3}
\begin{array}{lcl} \mathbb{E}(\Psi_k \Psi_l) &=& \mathbb{E}(\mathbb{E}(\Psi_k \Psi_l | \mathcal{R}^N_l))\\
&=& \mathbb{E}( \Psi_l \mathbb{E}(\Psi_k| \mathcal{R}^N_l))\\
&=& \mathbb{E}( \Psi_l \mathbb{E}( F^k_N -F^{k+1}_N| \mathcal{R}^N_l))\\
&=& \mathbb{E} ( \Psi_l \mathbb{E}( F^l_N -F^l_N))\\
&=& 0. \end{array}
\end{equation}
More precisely, in the first passage we exploit the fact that $\mathbb{E}(\mathbb{E}(\cdot |  \mathcal{R}_l^N)) = \mathbb{E}(\cdot)$, in the second passage we
exploit the fact that the conditioning is effective only on $\Psi_k$ (recalling that $k<l$), in the third passage we use the definition (\ref{eq:psi}), and in
the fourth passage we use Eq.~\ref{eq:martin} (recalling that $k+1 \leq l$).
\newline
Thus, in order to get the convergence in probability of the infinite volume limit of the free energy of the RBM (\ref{BMprima}) we only need to show that
\begin{equation}
\mathbb{E} \Psi_k^2 \le C,
\end{equation}
for some constant $C$. To this aim, let us define the following interpolating functions
\begin{equation}
\begin{array}{lcl}
H_k&=&-\frac{1}{\sqrt{N}} \sum_{i=1, i \not= k}^N \sum_{\mu=1}^P \xi^\mu_i \sigma_i z_\mu,\\
R_k&=&-\frac{1}{\sqrt{N}} \sum_{\mu=1}^P \xi^\mu_k \sigma_k z_\mu,\\
\Phi_k(t)&=& H_k+ t R_k, \\
g_k(t)&=& \log Z_{P,N}(\Phi_k(t))-\log Z_{P,N}(\Phi_k(0)).\\
\end{array}
\end{equation}
Notice that $H_k$ describes the original system but devoid of the $k$-th neuron (i.e., a ``cavity'' \cite{MPV}), $R_k$ represents the interaction term related to the $k$-th neuron, in such a way
that the interpolating function $\Phi_k(t)$ returns the original Hamiltonian when $t=1$, and $H_k$ when $t=0$. Finally, the function $g_k(t)$ vanishes when $t=0$,
while when $t=1$ it provides the contribution to the (extensive) free energy given by the $k$-th spin.
By applying the conditional average (\ref{eq:conditional}) to $g_k(1)$ and comparing with the definition (\ref{eq:psi}) we have the following identity
\begin{equation}
\Psi_k = \mathbb{E}_{<k} \left( g_k(1) \right) -\mathbb{E}_{<k+1} \left( g_k(1) \right).
\end{equation}
Next, by applying the Cauchy-Schwarz inequality and by averaging we get
\be \label{eq:passo1}
\mathbb{E} \Psi_k^2 \le 2 \mathbb{E}(g_k(1))^2.
\ee
Further, from the convexity of the partition function we have also
\begin{eqnarray} \label{eq:conv}
\frac{d^2 }{d t^2} g_k(t) \geq 0.
\end{eqnarray}
Exploiting the last expression, along with the fact that $g_k(0)=0$, we can get
\begin{equation}\label{eq:passo2}
g'_k(0) \leq g_k(1) \leq g'_k(1).
\end{equation}
The left inequality in (\ref{eq:passo2}) is obtained by a Taylor expansion of $g_k(1)$ around $t=0$, that is
$$g_k(1)=g_k(0) +g'_k(0) +g''_k(\xi)$$
where $0\le \xi \le 1$, from which it follows that $g_k(1)\ge g'_k(0)$ since
$g_k(0)=0$ and $g''_k(\xi) \geq 0$.
Analogously, for the right inequality in (\ref{eq:passo2}) we Taylor expand $g_k(0)$ around $t=1$, that is
$$g_k(0)=g_k(1)+(0-1)g_k'(1)+g_k''(\xi),$$
where $0\le\xi\le 1$, thus
$$g_k(1)=g_k'(1)-g''_k(\xi)$$
from which it follows that $g_k(1) \le g_k'(1)$ since $g_k''(\xi) \geq 0$.
\newline
In the following we will use the right inequality of (\ref{eq:passo2}) and, as for $g'_k(1)$, we notice that it is just the thermal average of the term $R_k$ and, exploiting the fact that the $\xi$ are uncorrelated and that $P/N$ is finite in the thermodynamic limit, one can state that
\begin{equation} \label{eq:bound}
\frac{\mathbb{E} \Psi_k^2}{2}  \le  \mathbb{E}\left(g_k(1)^2 \right) \le  \mathbb{E}\left(g_k'(1)^2\right) \le C.
\end{equation}
Now, recalling eqs.~(\ref{eq:flutt2})-(\ref{eq:flutt3}) and applying the bound in eq.~(\ref{eq:bound}), we finally get
\be
\mathbb{E}\left( A_{P,N}- \mathbb{E}\left(A_{P,N}\right)\right)^2 = \frac{1}{N^2} \sum_{i=1}^N
\mathbb{E} \Psi_k^2  \leq \frac{1}{N^2} \sum_{i=1}^N C = \frac{NC}{N^2} \rightarrow 0.
\ee
\end{proof}
\newline
\newline
It is worth mentioning that there are other techniques which may be used to prove Theorem 1, such as Talagrand's techniques on concentration of measure \cite{Tala4}; a clear outline on this perspective can be found in \cite{Bovier5}.

\section*{Appendix B. Proof of Proposition 2, Theorem 3 and Corollary 1}

\begin{proof}
The proof works in five stages. First ($i$) we evaluate explicitly the Cauchy condition $\phi_{P,N}(\beta, n,t=0) = \frac{1}{Nn} \ln \mathbb{E}\left( Z_0^n \right)$. Next ($ii$) we calculate $d \phi_{P,N}(\beta,n,t)/dt$, which we integrate back for $t\in[0,1]$; in this calculation ($iii$) we will explicitly make use of the assumption of replica symmetry. Then ($iv$), we construct the $n$-quenched replica symmetric free energy and we extremize the latter over the order parameters to get the related self-consistent equations. Finally ($v$), such equations are expanded to get the transition line.
\newline
\newline
{\em i. Cauchy condition.}
\begin{eqnarray}
\label{eq:Cauchy}
\phi_{P,N}(\beta, n,t=0)  &=& \frac{1}{Nn} \ln \mathbb{E} Z_0^n \\ \nonumber
   &=& \frac{1}{Nn} \ln \mathbb{E} \prod_{\gamma=1}^{n}\left[\sum_{\sigma_i^{\gamma}}\int_{-\infty}^{+\infty}\prod_{\mu=1}^{P} d\mu(z_{\mu}^{\gamma})
   e^{\sqrt{\alpha \beta \bar{p}} \sum_i^N J_i \sigma_i^{\gamma}+ \sqrt{\beta \bar{q}}  \sum_{\mu}^{P}J_{\mu}z_{\mu}^{\gamma} + \frac{\beta(1-\bar{q})}{2}\sum_{\mu}^{P}(z_{\mu}^\gamma)^2} \right]\\ \nonumber
   &=& \frac{1}{Nn} \ln \left\{ \mathbb{E}_{J_i}\sum_{\sigma_i^{\gamma}} e^{\sqrt{\alpha \beta \bar{p}} \sum_i^N J_i \sigma_i^{\gamma}}\cdot \mathbb{E}_{J_{\mu}}\int_{-\infty}^{+\infty}\prod_{\mu=1}^{P} d\mu(z_{\mu}^{\gamma}) e^{\sqrt{\beta \bar{q}} \sum_{\mu}^{P}J_{\mu}z_{\mu}^{\gamma} + \frac{\beta(1-\bar{q})}{2}\sum_{\mu}^{P}(z_{\mu}^\gamma)^2} \right\}\\
   &=& P_1 + P_2,
\end{eqnarray}
where
\begin{eqnarray}
P_1  &=& \ln 2 +   \frac{1}{n}\ln \int_{-\infty}^{+\infty}d\mu(x)\cosh^n \left(x \sqrt{\alpha \beta \bar{p}  } \right)   \\
P_2  &=& \frac{1}{Nn}\ln \mathbb{E}_{J_{\mu}}\left \{ \frac{1}{\left[ 1 - \beta (1-\bar{q}) \right]^{\frac{n P}{2}}}e^{\frac{\bar{q}\beta n P}{2\left [1-\beta(1-\bar{q}) \right]}J_{\mu}^2}  \right \} =  -\frac{\alpha}{2}\ln \left [1-\beta(1-\bar{q}) \right ] + \frac{\alpha \beta}{2} \frac{\bar{q}}{\left[1-\beta(1-\bar{q}) \right]}.
\end{eqnarray}
{\em ii. $t$-streaming of $\phi_{P,N}(\beta,n,t)$.}
\newline
The derivative of the $n$-quenched free energy reads as
\be
\dot{\phi}_{P,N}(\beta, n,t) \equiv \frac{d \phi_{P,N}(\beta, n,t)}{dt} = \frac{d}{dt} \frac{1}{Nn}\ln \mathbb{E}\left(Z^n_t \right) = \frac{1}{N}\frac{1}{\mathbb{E}\left(Z_t^n\right)}\mathbb{E}\left( Z_t^n \frac{\dot{Z}_t}{Z_t} \right)
\ee
and, in particular,
\begin{eqnarray}\nonumber
\frac{\dot{Z}_t}{Z_t} &\equiv& \frac{1}{Z_t}\frac{dZ_t}{dt}\\
\nonumber
&=& \frac{1}{Z_t}  \frac{d }{dt} \sum_{\sigma}\int_{-\infty}^{+\infty}\prod_{\mu}^{P}d \mu(z_{\mu}) e^{\sqrt{t}\left( \frac{\sqrt{\beta}}{\sqrt{N}}\sum_{i,\mu}^{N,P}\xi_i^{\mu}\sigma_i z_{\mu} \right) + \sqrt{1-t}\left( \sqrt{\alpha \beta \bar{p}} \sum_{i}^{N} J_i \sigma_i + \sqrt{\beta \bar{q}} \sum_{\mu}^{P} J_{\mu}z_{\mu}\right) + (1-t)\beta (1-\bar{q}) \sum_{\mu}^{P} \frac{z_{\mu}^2}{2}}  \\
&=& \tilde{\mathcal{A}}+  \tilde{\mathcal{B}} + \tilde{\mathcal{C}} + \tilde{\mathcal{D}},
\end{eqnarray}
where
\begin{eqnarray}
\tilde{\mathcal{A}} &=& \frac{\sqrt{\beta}}{2\sqrt{Nt}}\sum_{i,\mu}^{N,P}\xi_i^{\mu}\omega(\sigma_i z_{\mu}),\\
\tilde{\mathcal{B}} &=&  -\frac{\sqrt{\alpha \beta \bar{p}}}{2 \sqrt{1-t}}\sum_{i}^{N}J_i \omega(\sigma_i),\\
\tilde{\mathcal{C}} &=& -\frac{\sqrt{\beta \bar{q}}}{2 \sqrt{1-t}}\sum_{\mu}^{P}J_{\mu}\omega(z_{\mu}),\\
\tilde{\mathcal{D}} &=& - \frac{\beta}{2} (1-\bar{q})\sum_{\mu}^{P}z_{\mu}^2.
\end{eqnarray}
Thus, posing
\begin{eqnarray}
\mathcal{A} &=& \frac{1}{N}\frac{1}{\mathbb{E}\left(Z_t^n\right)}\mathbb{E}\left(Z^n_t \tilde{\mathcal{A}} \right),\\
\mathcal{B} &=& \frac{1}{N}\frac{1}{\mathbb{E}\left(Z_t^n\right)}\mathbb{E}\left(Z^n_t \tilde{\mathcal{B}} \right),\\
\mathcal{C} &=& \frac{1}{N}\frac{1}{\mathbb{E}\left(Z_t^n\right)}\mathbb{E}\left(Z^n_t  \tilde{\mathcal{C}} \right),\\
\mathcal{D} &=& \frac{1}{N}\frac{1}{\mathbb{E}\left(Z_t^n\right)}\mathbb{E}\left(Z^n_t \tilde{\mathcal{D}} \right).
\end{eqnarray}
we can recast $\dot{\phi}_{P,N}(\beta,n,t)$ as
\begin{equation} \label{eq:alltogether}
\dot{\phi}_{P,N}(\beta,n,t) = \mathcal{A} + \mathcal{B} + \mathcal{C} + \mathcal{D}.
\end{equation}
Let us start with the evaluation of $\mathcal{A}$:
\begin{eqnarray}
\mathcal{A} &=& \frac{1}{N}\frac{1}{\mathbb{E}\left( Z_t^n \right)}\mathbb{E} \left( Z^n_t \frac{\sqrt{\beta}}{2\sqrt{Nt}}\sum_{i,\mu}^{N,P} \xi_i^{\mu} \omega(\sigma_i z_{\mu})\right) \\
\label{eq:a1}
&=& \frac{1}{N}\frac{1}{\mathbb{E}\left( Z_t^n \right)} \frac{\sqrt{\beta}}{2\sqrt{tN}} \sum_{i,\mu}^{N,P} \mathbb{E} \xi_i^{\mu} Z_t^n \omega(\sigma_i z_{\mu})\\
&=& \frac{1}{N}\frac{1}{\mathbb{E}\left( Z_t^n \right)} \frac{\sqrt{\beta}}{2\sqrt{tN}} \sum_{i,\mu}^{N,P} \mathbb{E}
\label{eq:a2}
\partial_{\xi_i^{\mu}} \left( Z_t^n \omega(\sigma_i z_{\mu}) \right)\\
&=& \frac{1}{N}\frac{1}{\mathbb{E}\left( Z_t^n \right)} \frac{\sqrt{\beta}}{2\sqrt{tN}} \sum_{i,\mu}^{N,P} \mathbb{E} \left( n Z_t^{(n-1)} \frac{d Z_t}{d\xi_i^{\mu}} \omega(\sigma_i z_{\mu}) + Z_t^n \partial_{\xi_i^{\mu}} \omega(\sigma_i z_{\mu}) \right)\\
&=& \mathcal{A}_1 + \mathcal{A}_2,
\end{eqnarray}
where
\begin{eqnarray}
\mathcal{A}_1 &=& n \cdot \frac{P \beta}{2N} \langle q_{12} p_{12}\rangle_n,\\
\mathcal{A}_2 &=& \frac{\beta}{2N}\sum_{\mu}^{P}\langle z_{\mu}^2 \rangle_n - \frac{P \beta}{2N}\langle q_{12}p_{12}\rangle_n.
\end{eqnarray}
Notice that to get from (\ref{eq:a1}) to (\ref{eq:a2}) we applied Wick's theorem, namely $\mathbb{E}[x_i F(x)] = \sum_k A_{ik}^{-1} \mathbb{E}[\partial F(x) / \partial x_k]$, where $A_{ik}^{-1}= \mathbb{E}(x_i x_k)$, and, since our variables are i.i.d. from $\mathcal{N}(0,1)$, we have $A_{ik}^{-1}=\delta_{ik}$.\\
Performing analogous calculations overall we obtain
\begin{eqnarray}
\mathcal{A} &=& (n-1) \frac{\alpha \beta}{2}\langle q_{12}p_{12}\rangle_n + \frac{\alpha \beta}{2}\frac{1}{P}\sum_{\mu}^{P}\langle z_{\mu}^2\rangle_n ,\\
\mathcal{B} &=& -(n-1)\frac{\alpha \beta}{2}\bar{p}\langle q_{12}\rangle_n - \frac{\alpha \beta}{2}\bar{p},\\
\mathcal{C} &=& -(n-1)\frac{\alpha \beta}{2}\bar{q}\langle p_{12}\rangle_n  - \frac{\alpha \beta}{2}\bar{q}\frac{1}{P}\sum_{\mu}^{P}\langle z_{\mu}^2\rangle_n,\\
\mathcal{D} &=& - \frac{\alpha \beta}{2}(1-\bar{q})\frac{1}{P}\sum_{\mu}^{P}\langle z_{\mu}^2\rangle_n.
\end{eqnarray}
Putting all the terms together according to (\ref{eq:alltogether}) and noticing that the nasty terms containing the factor $\frac{1}{P}\sum_{\mu}^{P}\langle z_{\mu}^2\rangle_n$ cancel out, we have
\be
\dot{\phi}_{P,N}(\beta, n,t)= (n-1)\frac{\alpha \beta}{2}\langle q_{12}p_{12}\rangle_n-(n-1)\frac{\alpha \beta}{2}\bar{p}\langle q_{12}\rangle_n - \frac{\alpha \beta}{2}\bar{p}-(n-1)\frac{\alpha \beta}{2}\bar{q}\langle p_{12}\rangle_n.
\ee
{\em iii. Integration.}\\
Recalling that, in the replica symmetric regime and in the thermodynamic limit,
$$
0 = \langle \left(q_{12}-\bar{q}\right)\left(p_{12}-\bar{p}\right)\rangle_n,
$$
we can use the above relation to verify that
\be
\lim_{N \to \infty} \dot{\phi}_{P,N}(\beta, n,t) = -\frac{\alpha \beta}{2} \bar{p} \left[1 + (n-1)\bar{q} \right].
\ee
Its integration back in $t$ simply coincides with the multiplication by one:
\be\label{sumruleintegrata}
\lim_{N \to \infty} \int_{0}^{1} dt \dot{\phi}_{P,N}(\beta,n,t) = -\frac{\alpha \beta}{2} \bar{p} \left[1 + (n-1)\bar{q} \right]\cdot 1,
\ee
and this closes the calculations of the various contributions to $\dot{\phi}_{P,N}(\beta,n)$, in the asymptotic limit.
\newline
\newline
{\em iv. Extremization.}
\newline
Recalling eqs.~(\ref{eq:defphi}), (\ref{Z-interpolante}) and (\ref{eq:phibetant}),  $\phi_{P,N}(\beta,n, t=1) = (Nn)^{-1} \log \mathbb{E}(Z^n) = \phi_{P,N}(\beta,n)$ and, by a trivial application of the Fundamental Theorem of Calculus, $\phi_{P,N}(\beta,n, t=1)  = \phi_{P,N}(\beta,n, t=0) + \int_0^1 \dot{\phi}_{P,N}(\beta,n,t) dt$.
In particular, in the thermodynamic limit, summing the Cauchy condition (\ref{eq:Cauchy}) to the r.h.s. of eq. (\ref{sumruleintegrata}) we obtain $\lim_{N \to \infty}\phi_{P,N}(\beta, n,t=1)=\lim_{N \to \infty}\phi_{P,N}(\beta,n)$ as
\begin{eqnarray}
\nonumber
\lim_{N \to \infty}\phi_{P,N}(\beta,n) &=& \ln 2 +   \frac{1}{n}\ln \int_{-\infty}^{+\infty}d\mu(x)\cosh^n \left(x  \sqrt{\alpha \beta \bar{p} } \right)
 -\frac{\alpha}{2}\ln \left[ 1-\beta(1-\bar{q}) \right] \\
\nonumber
 &+& \frac{\alpha \beta}{2} \frac{\bar{q}}{1-\beta(1-\bar{q}) } -\frac{\alpha \beta}{2} \bar{p} \left[1 + (n-1)\bar{q} \right].
\end{eqnarray}
whose extremization with respect to $\bar{q},\ \bar{p}$ returns
\begin{eqnarray}
\label{app:sc1}
\frac{\partial \phi_{P,N}(\beta,n)}{\partial \bar{q}} &=& 0 \Rightarrow \bar{p} = \frac{ \beta \bar{q}}{(1-n)\left[1-\beta \left(1-\bar{q}\right)\right]^2}, \\
\label{app:sc2}
\frac{\partial \phi_{P,N}(\beta,n)}{\partial \bar{p}} &=& 0 \Rightarrow \bar{q} = \int d \mu(x) \tanh^2\left(x \sqrt{\alpha \beta  \bar{p}}  \right).
\end{eqnarray}
The solution of these self-consistent equations provide the expectation of $q$ and $p$ in the thermodynamic limit and under the replica symmetry.
\end{proof}
\newline
\newline
{\em v. Transition line.}
We expand (\ref{app:sc1}) and (\ref{app:sc2}) around $\bar{p} = \bar{q} =0$ and, combining the two equations we get
\begin{equation}
\frac{\alpha \beta^2}{(1 - \beta)^2 (1-n)}=1 \Rightarrow (T-1)^2 = \frac{\alpha}{1-n},
\end{equation}
and, with some algebra one recovers (\ref{eq:criticalline}).

\end{document}